\documentclass[sigplan]{acmart}

\AtBeginDocument{%
  \providecommand\BibTeX{{%
    \normalfont B\kern-0.5em{\scshape i\kern-0.25em b}\kern-0.8em\TeX}}}

\acmConference[]{}{}{} %[LICS'22]{}{}{}
\acmBooktitle{} %{LICS'22}
\acmPrice{}
\acmISBN{}

\date{}
\acmConference{}{}{} %{Logic in Computer Science}{2-5 August 2022}{Haifa, Israel}
\acmBooktitle{}

\setcopyright{none}
\copyrightyear{}
\acmYear{}
\acmDOI{}
\acmISBN{}

\title{Active learning sound negotiations}

\author{Anca Muscholl}
\orcid{nnnn-nnnn-nnnn-nnnn}             %% \orcid is optional
\affiliation{
  \department{LaBRI}              %% \department is recommended
  \institution{Bordeaux University}            %% \institution is required
  \country{France}                    %% \country is recommended
}
\email{}          %% \email is recommended

\author{Igor Walukiewicz}
\orcid{nnnn-nnnn-nnnn-nnnn}             %% \orcid is optional
\affiliation{
  \department{LaBRI}             %% \department is recommended
  \institution{CNRS, Bordeaux University}           %% \institution is required
  \country{France}
}
\email{}         %% \email is recommended

\renewcommand{\shortauthors}{}

\begin{abstract}
  We present two active learning algorithms for sound deterministic
  negotiations. 
  Sound deterministic negotiations are models of distributed systems, a kind of 
  Petri nets or Zielonka automata with additional structure.
  We show that this additional structure allows to minimize such negotiations.
  The two active learning algorithms differ in the type of membership queries they use.
  Both have similar complexity to Angluin's $L^*$ algorithm, in
  particular, the number of queries is polynomial in the size of the
  negotiation, and not in the  number of configurations. 
\end{abstract}

\usepackage{booktabs}   %% For formal tables:
\usepackage{subcaption} %% For complex figures with subfigures/subcaptions
\usepackage{hyperref}
\usepackage{latexsym}
\usepackage{amsmath,amsthm} 
\theoremstyle{definition}
\usepackage{logic-thm}
\usepackage{tikz}
\usepackage{stackengine}
\usepackage{graphicx}
\usetikzlibrary{positioning}
\usepackage[disable]{todonotes}
\usepackage{xcolor}
\usepackage[]{algorithm2e}

\newcommand{\anca}[1]{\todo[color=blue!30]{\small #1}}

\newcommand{\igw}[1]{\todo[color=green!30]{\small #1}}

\newtheorem{theorem}{Theorem}[section]
\newtheorem{proposition}{Proposition}[theorem] 
\newtheorem{corollary}{Corollary}[theorem]
\newtheorem{lemma}[theorem]{Lemma}
\newtheorem{definition}[theorem]{Definition}

\newtheorem{remark}[theorem]{Remark}

\newcommand{\Adom}{A_{\dom}}

\newcommand{\Teacher}{{Teacher}}
\newcommand{\Learner}{{Learner}}
\newcommand{\Unique}{\emph{Uniqueness}}
\newcommand{\Closure}{\emph{Closure}}
\newcommand{\NEQ}{\emph{Node-mismatch}} %{\emph{Non-equivalence}}
\newcommand{\OUTINC}{\emph{Absent-trans}}
\newcommand{\Target}{\emph{Target-mismatch}} %{\emph{Target-incomplete}}
\newcommand{\PREF}{\emph{Pref}}
\newcommand{\DOM}{\emph{Domain}}

\newcommand{\EQ}{\mathit{EquivQuery}}

\newcommand{\res}{\mathit{res}}
\newcommand{\larr}{\leftarrow}
\newcommand{\OUT}{\mathit{OUT}}
\newcommand{\CLOS}{\mathit{CLOS}}
\newcommand{\BinS}{\mathit{BinS}}
\newcommand{\TRG}{\mathit{TRG}}
\newcommand{\Negotiation}{\mathit{Negotiation}}

\renewcommand{\d}{\delta}

\newcommand{\tequiv}{\approx}

\newcommand{\init}{\mathit{init}}

\newcommand{\Cinit}{C_{\init}}
\newcommand{\Cfin}{C_{\fin}}
\newcommand{\ninit}{n_{\init}}
\newcommand{\nfin}{n_{\fin}}

\newcommand{\dact}[1]{\stackrel{#1}{\longmapsto}}
\newcommand{\ans}{\mathit{ans}}
\newcommand{\isSound}{\mathit{IsSound}}

\newcommand{\eqL}{\equiv^L}
\newcommand{\eqT}{\equiv_T}
\newcommand{\eqTp}{\equiv_{T'}}

\newcommand{\out}{\mathit{out}}
\newcommand{\dmin}{\mathit{dmin}}
\newcommand{\ndom}{\mathit{dnode}}
\renewcommand{\dom}{\mathit{dom}}

\newcommand{\Proc}{\mathit{Proc}}

 \newcommand{\paths}{\mathit{Paths}}

 \newcommand{\sqs}{\preccurlyeq}

\newcommand{\tNn}{\widetilde{\Nn}}
\newcommand{\tL}{\widetilde{L}}
\newcommand{\tC}{\widetilde{C}}
\newcommand{\tCinit}{\widetilde{C}_{\init}}
\newcommand{\tAa}{\widetilde{\Aa}}

\newcommand{\actionformat}[1]{\mathsf{#1}}
\newcommand{\aappl}{\actionformat{appl}}
\newcommand{\ainfo}{\actionformat{info}}
\newcommand{\asetup}{\actionformat{setup}}
\newcommand{\adinit}{\actionformat{dinit}}

\newcommand{\asvote}{\actionformat{svote}}

\newcommand{\avote}{\actionformat{vote}}
\newcommand{\adec}{\actionformat{dec}}
\newcommand{\aack}{\actionformat{fin}}

\begin{document}

%% Title information
\title[Active learning sound negotiations]{Active Learning Sound Negotiations}         %% [Short Title] is optional;
  %% \keywords are mandatory in final camera-ready submission

%% \maketitle
%% Note: \maketitle command must come after title commands, author
%% commands, abstract environment, Computing Classification System
%% environment and commands, and keywords command.
\maketitle

\section{Introduction}

The active learning paradigm proposed by Angluin~\cite{ang87} is a method used by a
\Learner\ to identify an
unknown language.
The paradigm assumes the existence of a \Teacher\ who can answer
membership and equivalence queries.
Learner can ask if a word belongs to the language being learned, or if an
automaton she constructed accepts that language. 
This setting allows for much more efficient algorithms than passive learning,
where Learner receives just a set of positive and negative examples~\cite{Hig10}.
While passive learning has high theoretical
complexity~\cite{TraBar:73,gold78}, Angluin's  $L^*$-algorithm
can learn a regular language  with polynomially many queries to  \Teacher.
Active learning algorithms have been designed for many extensions of
deterministic finite automata: automata on infinite words, on trees, weighted
automata, nominal automata, bi-monoids for pomset
languages~\cite{dh07,bm15,mw15,af16,chjs16,mssks17,om20,HKRS21}.
Following Angluin's original algorithm, several algorithmic
improvements have been proposed~\cite{RS93,KV94,TTT14}, implemented in learning
tools~\cite{learnlib,BolKatKer10b}, and used in case
studies~\cite{vaa17,SmeenkMVJ15,RuiPol15,FitJanVaa16,NeiderSVK97,TapAicBlo19}. 
%Generic categorical frameworks for learning covering even more examples have
%been proposed recently. 

Learning distributed systems is a particularly promising direction. 
First, because most systems are distributed anyway. 
Second, because distributed systems exhibit the state explosion phenomenon,
namely, the state space of a distributed system is often exponential
in the size of the description of the system.
If we could learn a distributed system in time polynomial in the size of
the description, we would be using state explosion to our advantage. 
Put differently, knowing something about the structure of the system would allow 
to speed up the learning process exponentially. 

The learning results cited above all rely  on the existence of canonical automata,
even though sometimes these automata may not be minimal. 
% Having a minimal canonical object is even better, as the complexity of an
% algorithm depends on the size of the object being constructed.
This is a main obstacle for learning distributed systems.
Consider the following example that can be reproduced in many kinds of systems.
Suppose we have two processes, $p_1$ and $p_2$, both executing a shared action
$b$. 
It means that on executing $b$ the two processes update their state.
The goal of the two processes is to test if the number of actions $b$ is a
multiple of $15$. 
One solution is to make $p_1$ count modulo $3$ and $p_2$ to count modulo $5$. 
Each time when the two remainders are $0$ they can declare that the number of
$b$'s they have seen is divisible by $15$.
The sum of the number of states of the two processes is $3+5=8$.
Another possibility is that $p_1$ stores the two lower bits of count
modulo $15$, 
and $p_2$ stores the two higher bits. 
The sum of the number of states of the two processes is $4+4=8$. 
It is clear that there is no distributed system for this language with $2+5$ states or with $3+4$ states, as
the number of global states would be $2*5=10$ and $3*4=12$, respectively. 
Thus we have two non-isomorphic minimal solutions.
But it is not clear which of the two should be considered canonical.
It is hard to imagine a learning procedure that would somehow chose one solution
over the other. 
In this paper we avoid this major obstacle.
The distributed automata we learn, \emph{sound deterministic negotiations},  cannot implement any of the two solutions. 
The minimal solution for negotiations has $15$ nodes and resembles the minimal
deterministic automaton for the language. 

Negotiations are a distributed model proposed by Esparza and Desel in~\cite{ED13},
tightly related to workflow nets~\cite{Aal:16} and 
free-choice Petri nets.
In one sentence, this model is a graph-based representation of processes
synchronizing over shared actions.
Figure~\ref{fig:lipics} shows a negotiation corresponding to
the workflow of an editorial board, with 4 processes
$NA$ (new application), $TS$ (technical support), $EC$ (editorial board chair), $EM$ (editorial board member).
Actions are written in blue, for instance {\color{blue} $\textsf{svote}$}
(set-up vote) is
a shared action of processes $EC$ and $TS$.
At node $n3$ processes $TS,EC$ have the choice between actions
{\color{blue} $\textsf{svote}$} and {\color{blue} $\textsf{tech}$}.
Taking jointly {\color{blue} $\textsf{svote}$} leads process $TS$ to $n6$ and
$EC$ to $n5$.
The semantics of a negotiation is a set of \emph{executions}, namely sequences of
actions that are executable from an initial to a final state.
In our example,
$(\aappl)(\asetup)(\adinit)(\aack)(\asvote)(\avote)(\adec)$ is an execution.
Executions are  Mazurkiewicz traces~\cite{maz77} because there is a natural independence relation
between actions: if the domains of two actions are disjoint, the actions  are
independent, and otherwise not.
\begin{figure}
   \includegraphics[width=.4\textwidth,height=.3\textheight]{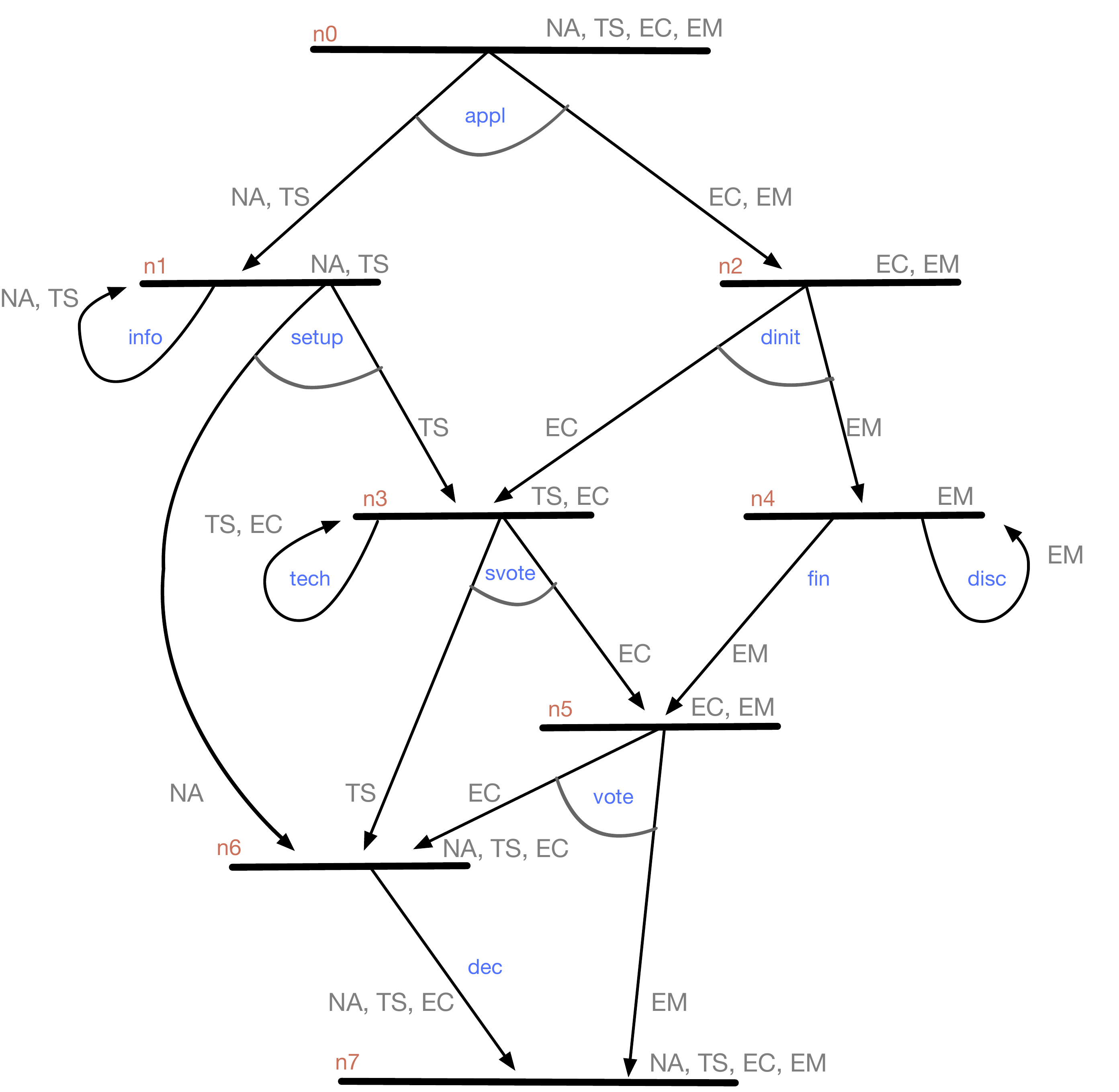}
	\caption{A sound, deterministic negotiation}
	\Description{none}
	\label{fig:lipics}
\end{figure}
      
Negotiations that are deterministic and sound,
as the one in Figure~\ref{fig:lipics},  turn out to have a close
relationship with finite automata.
Soundness is a variant of deadlock-freedom,
and determinism means that every state has at most one outgoing transition on a
given label. 
Our first result is a canonical representation for sound deterministic
negotiations by finite automata, that also provides a minimization
result.

Based on this canonical representation, one could just use the standard
Angluin algorithm $L^*$ for DFA to learn sound, deterministic
negotiations in polynomial time. 
This results in a rather unrealistic setting where \Teacher\ is supposed to have 
access to the graph representation of a negotiation.
When learning the negotiation from Figure~\ref{fig:lipics}, this setting would
e.g.~require \Teacher\ to answer with a \emph{local path} in the graph, like for
example the leftmost path $(\aappl_{TS})(\asetup_{NA})(\adec_{EC})$ from $n0$ to
$n7$.
However, if the negotiation under learning is  black-box, then equivalence
queries need to be
approximated by conformance testing~\cite{vaa17}.
In this case local paths are not accessible to \Teacher: he can only
apply executions to the system under learning.
Therefore  we assume in this paper that when the two negotiations are not
equivalent \Teacher\ replies with a counter-example in form of an execution that belongs to one negotiation
but not to the other.

% {\color{gray}
% Based on this canonical representation, one could just use the standard
% Angluin algorithm for DFA to learn sound, deterministic negotiations. 
% This would be meaningful if the setting would allow to access the finite
% automaton representation directly.
% In our case the representation talks about the local paths in the negotiation graph,
% like the leftmost path $\aappl_{T}\asetup_C\adec_{EC}$ from $n0$ to
% $n7$ in the graph in Figure~\ref{fig:lipics}.
% We find it difficult to imagine a usage scenario where  \Teacher\ can reply to
% equivalence queries by local paths of the negotiation.
% For example, when the system under learning is black-box, equivalence
% queries need to be
% approximated by conformance testing~\cite{vaa17}. 
% So it is realistic to assume that  if the
% two negotiations are not equivalent, \Teacher\ replies with  an execution that
% belongs to one negotiation but not the other.
% % In our example, an execution is a sequence of actions, like
% %$(\aappl)(\asetup)(\adinit)(\aack)(\asvote)(\avote)(\adec)$.
% Of course, one cannot infer from an execution any details about the
% internals of the negotiation, so about local paths associated with the
% execution.
% }

As \Teacher\ replies with executions to equivalence queries,
the main challenge is to extract some information from a counter-example
execution allowing to extend the negotiation under learning. 
In our first algorithm \Learner\ can ask membership queries about local
paths. 
Membership queries about local paths are arguably difficult to justify,
yet the algorithm is relatively simple and serves as a basis for the second algorithm.

Our second learning algorithm uses only executions, both for membership
and for equivalence queries. 
With a counter-example at hand, \Learner\ needs to be able to find a
place to modify the negotiation she constructed so far.
For this the negotiation needs to have enough structure to allow
to build executions for membership queries. 
Even though this induces an important conceptual complication, the complexity of
our second algorithm is comparable to that of the standard $L^*$ algorithm
for DFA.
Moreover, equivalence queries in this algorithm can be done in \PTIME, if the negotiation to learn is given explicitly
to \Teacher.

% \paragraph{Summary of the results:}
% \begin{itemize}
% 	\item A minimization result for sound deterministic negotiations. The
% 	minimal negotiation can be constructed from the minimal automaton 
% 	recognizing the set of local paths of a negotiation. 
% 	\item An active learning algorithm where \Learner\ can ask membership
% 	queries about local paths, and \Teacher\ replies with executions to equivalence
% 	queries.
% 	\item An active learning algorithm where both membership and  equivalence
% 	queries are about executions. 
% 	\item The complexity of both learning algorithms is polynomial in the size of the
% 	negotiation being learned, and roughly the same as the complexity of Angluin's
% $L^*$	algorithm for learning DFAs. 
% \end{itemize}

\paragraph{Related work}
The active learning paradigm was initially designed for regular languages~\cite{ang87}. 
It is still the basis of all other learning algorithms. 
From the optimizations proposed in the literature~\cite{RS93,KV94,BolKatKer10b,TTT14} we
adopt two in this work. We use discriminator trees instead of rows, as this
allows to gain a linear factor on the number of membership queries. 
We also use  binary search to find a place where a modification should be made. 
This gives a reduction from $m$ to $\log(m)$ membership queries to process a
counterexample of size $m$. 
As it is also common by now, we add only those suffixes from a counter-example
that are needed to create new states or transitions.
These and some other optimizations are implemented in the
TTT-algorithm~\cite{TTT14}. 

There are many extensions of the active learning setting to richer models:
$\w$-regular languages, weighted languages, nominal languages, tree languages,
series-parallel pomsets~\cite{dh07,bm15,mw15,af16,chjs16,mssks17}.
All of them rely on the existence of a canonical automaton for a given language.
The algorithm for learning non-deterministic automata is not an exception as it
learns residual finite state automata. %  (an actually deterministic finite automata
% over join-semilattices~\cite{US20}). 
% Similarly, the algorithm for learning series parallel pomsets, learns
% deterministic tree automata. 
% In our case we learn just deterministic finite automata representing
% negotiations.
Categorical frameworks have been recently proposed to cover the majority of these examples
and provide new ones~\cite{HSS17,UrbSch20,CPS21}. 

To our knowledge the first active learning algorithm for  concurrent
models is~\cite{BKKL10}, where message-passing automata are learned
from MSC scenarios.
However, this algorithm requires a number of queries
that is exponential in the number of processes and the channel bounds.
Recently, an active learning algorithm for series-parallel pomsets was
proposed~\cite{HKRS21}.
This algorithm learns bimonoids recognizing
series-parallel pomsets, which may be exponentially larger than a
pomset automaton accepting the language.
It relies on  a representation of series-parallel pomsets as trees, 
and learns a tree automaton accepting the set of representations.
% It is not clear if this learning algorithm is an instance of a general
% categorical setting~\cite{CPS20}.
Note that languages of deterministic sound negotiations and of series-parallel
pomsets are incomparable.
For example, the pomsets corresponding to executions of the negotiation from
Figure~\ref{fig:lipics} are not series-parallel.

Negotiations have been proposed by 
Esparza and Desel~\cite{ED13,DEH19}. 
It is a model inspired by workflow nets~\cite{Aal98,Aal:16} but using 
processes like in Mazurkiewicz trace theory and Zielonka automata~\cite{maz77,zie87,dr95}.
Workflow nets have been studied extensively, in particular variants of black-box
learning~\cite{AalCarCha19}, but we are not aware of any result about active learning of such nets.
% making a link between workflow nets and finite automata as the one we present
% here.

% There is a quite rich literature on learning workflow nets from logs,
% that is a \anca{why should we talk about this?}
% variant of black-box learning~\cite{AlstSurvey}.
% First, there is an approach is based on theory of regions in Petri
% nets~\cite{EF90,Carmona09,Aalst09}.
% Each action must represent exactly one transition in a discovered Petri net,
% moreover the resulting net may be exponential in the size of a transition system
% representation. 
% Another approach construct a series-parallel representation of sequences in the
% log\. 
% The method is much faster but does not give guarantees on the quality of approximation~\cite{Leemans2013}.
% Talk about learning from workflows from logs.

\paragraph{Structure of the paper.}
In the next section we give an overview and the context of the paper.
In Section~\ref{sec:basics} we define sound, deterministic negotiations.
Section~\ref{sec:minimization} presents the result on minimization.
Section~\ref{sec:angluin} recalls briefly Angluin's $L^*$ algorithm. 
% It is a slight variation on existing presentations as we prefer to work with
% finite automata that are not complete. 
Sections~\ref{sec:paths}, and~\ref{sec:executions} describe the two learning
algorithms that are the main result of the paper.
Omitted proofs can be found in the Appendix.

%%% Local Variables:
%%% mode: latex
%%% TeX-master: "m"
%%% End:

\section{Overview}\label{sec:overview}
Before going into the technical content of our work we give a high-level
overview of the key concepts and results.

A negotiation is like a finite automaton with many tokens.
The behavior of a finite automaton can be described in terms of one token
moving between states, that we prefer to call nodes, in the graph of the automaton.
At first, the token is in the initial node. 
It can then take any transition outgoing from this node and move further.
If the transition is labelled by $b$, we say that the automaton takes action $b$. 
With this view, words accepted by the automaton are sequences of actions
leading the token from the initial node to a final one.

What happens if we put two tokens in the initial node?
When we look at the sequences of actions that are taken we will get a shuffle of words in
the language of the automaton. 
This is concurrency without any synchronization.

Negotiations are like finite automata with several tokens and a very
simple synchronization mechanism. 
The number of tokens is fixed and each of them is called a process,
say from a finite set $\Proc$.
The processes move from one node to another
according to the synchronization mechanism 
described in the following.
Every node has its (non-empty) domain $\ndom: N\to 2^\Proc$ and a set
of outgoing actions.
The node's domain says which processes can
reach it: process $p$ can reach only nodes $n$ with $p\in\ndom(n)$. 
The synchronization requirement is that \emph{all} processes in
$\ndom(n)$  leave node $n$ jointly, after choosing a common outgoing
action. % only if \emph{all} processes from
% $\ndom(n)$ are in $n$.
% Moreover all processes from $\ndom(n)$ should move
% at the same time and on the same action.
Taking the same action at node $n$  means that processes from
$\ndom(n)$ ``negotiate''  which action they take jointly.
As in the case of finite automata, an \emph{execution} in a negotiation is
determined by a sequence of actions labelling the transitions  taken, except that
now one action corresponds to a move of potentially several processes. 
The non-deterministic variant of this simple mechanism can simulate $1$-safe Petri Nets or Zielonka automata,
albeit with many deadlocks. 

Recall the negotiation in Figure~\ref{fig:lipics}, with the four
processes $\Proc=\set{NA,TS,EC,EM}$. %  gives an example negotiation, corresponding to
% the workflow of an editorial board.\igw{Do we talk about LIPICS?}
Nodes are represented by horizontal bars. 
The initial node is on the top, and the final one at the bottom. 
The domain of every node, $\ndom(n)$, is indicated just above the node to the right.
Actions are written in blue, with $\Act=\set{\aappl,\asetup,\dots,\adec}$. 
From every node there are several outgoing transitions on the same action, one
transition per process in the domain of the node.
For example, from the initial node there is an action $\aappl$ with four
transitions, one for each process.
We denote by $\aappl_{TS}$ the transition labelled $\aappl$ of process $TS$. 
Transition $\aappl_{TS}$ leads $TS$ from $n0$ to $n1$. 
Node $n1$  has two outgoing transitions, $\ainfo$ and $\asetup$. 
Both involve the two processes $NA,TS$. 
Every transition from node $n$ involves all processes in the domain of $n$.

% \begin{figure}
% 	\includegraphics[scale=.3]{lipics.pdf}
% 	\caption{Example negotiation}
% 	\label{fig:lipics}
% \end{figure}

All processes start in the initial node $n0$. 
After action $\aappl$ processes $NA,TS$ reach node $n1$, from where they can take
action $\asetup$ leading $TS$ to node $n3$, and $NA$ to node $n6$. 
In parallel processes $EC,EM$ reach node $n2$ from where they can take action $\adinit$, which
makes $EC$ rejoin $TS$ in node $n3$.
They can continue like this forming an execution from $n0$ to $n7$:
%\[
	$(\aappl)(\asetup)(\adinit)(\aack)(\asvote)(\avote)(\adec)$.
%\]
Observe that the order of $\asetup$ and $\adinit$ is not relevant because they
appear concurrently. 
We say that the two actions are independent because they have disjoint domains. 
On the other hand $\adinit$ and  $\asvote$ cannot be permuted because $EC$ is in
the domain of the two actions. 
Actions are therefore partially ordered in an execution. 
We write $L(\Nn)\incl \Act^*$ for the set of all (complete) executions
of negotiation $\Nn$.

More formally, actions in a negotiation are typed forming a \emph{distributed alphabet}. 
Every action is assigned a set of processes participating in that action:
$\dom:\Act\to2^\Proc$. 
Going back to our example from Figure~\ref{fig:lipics}:
$\dom(\aappl)$ is the set of all four processes, while 
$\dom(\asetup)=\set{NA,TS}$ and $\dom(\adinit)=\set{EC,EM}$.
For every node $n$ and action $a$ outgoing from $n$ we have
$\dom(a)=\ndom(n)$.
This way executions of negotiations can be viewed as
Mazurkiewicz traces~\cite{maz77}.
As the domains of $\asetup, \adinit$ are disjoint the two actions are independent, so their order can be
permuted: for all  $u,v \in \Act^*$, $u(\asetup)(\adinit)v \in L(\Nn)$ iff $u(\adinit)(\asetup)v \in
L(\Nn)$.

Two negotiations $\Nn_1$, $\Nn_2$ over the same distributed alphabet are \emph{equivalent}
if $L(\Nn_1)=L(\Nn_2)$.  
Since we will consider negotiations without  deadlocks, and our systems are
deterministic, this is equivalent to the two negotiations being strongly bisimilar.
The goal of active learning is to allow Learner to find a negotiation
equivalent to the one known by Teacher, assuming Learner can ask membership and
equivalence queries to Teacher.

\paragraph{Sound, deterministic negotiations.}
Negotiations can simulate Petri nets or Zielonka automata.
The three models suffer from the main obstacle described in the introduction.
For deterministic negotiations this changes when we impose soundness.
A negotiation is \emph{sound}, if every execution starting from the initial node
can be extended to an execution that reaches a final node.
(Without loss of generality we will assume that there is only one final node in a
negotiation.)
So soundness is a variant of deadlock freedom.
A negotiation is \emph{deterministic} if for every process $p$ and
action $b$ every node has at most one outgoing edge labeled $b$ and leading to
a node with $p$ in its domain.
The negotiation from Figure~\ref{fig:lipics} is sound and deterministic.

Sound deterministic negotiations have many interesting properties.
While soundness looks like a semantic property, it can be decided in \NLOGSPACE\
for deterministic negotiations~\cite{EKMW18}. 
Actually, soundness is characterized by forbidden patterns in the
negotiation graph.
Some quantitative properties of sound deterministic negotiations can be computed in
\PTIME, see \cite{EspMusWal:17}.
  But not everything is easy.
Deciding if a given negotiation has some execution that
belongs to a given regular language is \PSPACE-complete.

\paragraph{Our results.}
Our first contribution  is the observation that sound deterministic
negotiations can be minimized. 
This presents prospects for  Angluin-style learning, as there is a
canonical object to learn.
It also provides a simple polynomial-time equivalence algorithm for
such negotiations.

To explain the minimization result, we need one more notion. 
A \emph{local path} in a negotiation is a labelled path in
the negotiation graph, for example $(\aappl_{TS})(\asetup_{NA})(\adec_{EC})$ in the
negotiation from Figure~\ref{fig:lipics}.
Since the negotiation is deterministic, the source node, the action,
and the process uniquely determine the transition. 
We write $\aappl_{TS}$ for the transition on $\aappl$ of process $TS$.
In general local paths are sequences over the alphabet $A_{\dom}=\set{a_p:
  a \in \Act,p\in\dom(a)}$.
We write $\paths(\Nn)$ for the set of all local paths of
$\Nn$
leading from the initial to the final node.

Negotiations can be minimized by simply minimizing the finite automaton for
local paths, Proposition~\ref{prop:mindfa-negotiation}.
% \noindent\textbf{Proposition [Prop.~\ref{prop:mindfa-negotiation}]:}
% %\begin{proposition}\igw{copy of}
%   Let $\Nn$ be a sound deterministic
%   negotiation, and let $\Aa$ be the minimal DFA accepting
%   $\paths(\Nn)$. From $\Aa$ we can construct a sound deterministic negotiation
%   $\Nn_\Aa$ such that $L(\Nn)=L(\Nn_\Aa)$.
%   %\end{proposition}
% \smallskip
This proposition suggests using Angluin-style learning
for finite automata to learn sound negotiations. 
But this supposes that Learner asks questions about local paths, and Teacher
replies with local paths as counter-examples. 
As already mentioned, we find it hard to justify this setting. % , for instance when the negotiation
% under learning is accessible as black-box.
Instead, we consider the scenario where Teacher replies with a complete
execution (and not a local path).

Our first learning algorithm, Theorem~\ref{thm:learning-paths}, still allows
\Learner\ to ask membership queries about local paths. 
Admittedly, this may be not very realistic either, 
but the algorithm is instructive, % is quite similar to Angluin's algorithm for DFA, but at the same time
% it
using some concepts that are central for our second algorithm. 
The main challenge is how to extract from a counter-example given by \Teacher\
some information allowing to modify a negotiation being learned. 
The crucial property is that when \Learner\ runs a counter-example given by \Teacher\ in a
negotiation being learned then she can find an inconsistency in her information
before the counter-example reaches a deadlock (Lemma~\ref{lem:lp-pos}).
%\smallskip
% %\begin{theorem}
% %  For a given distributed alphabet with $|\Proc|$ processes.
%  \noindent\textbf{Theorem [Thm~\ref{thm:learning-paths}]:} There is an active learning
%  algorithm for  sound deterministic negotiations, 
%   using membership que\-ries on local paths and equivalence queries returning
%   executions. It can learn a negotiation of size $s$ using $O(s(s+
%   |\Proc|+\log(m)))$ membership queries and $s$ equivalence queries,
%   with $m$ the size of the longest counter-example.
% %\end{theorem}
% \smallskip

In our second, main learning algorithm Learner can ask membership queries about
executions, and not about local paths, Theorem~\ref{thm:learning-executions}.
The challenge now is how to construct membership queries about executions, and how
to extract useful information from the answers. 
In the first algorithm membership queries about local paths allowed to
obtain information about the graph of the negotiation. 
It is not evident how to use executions to accomplish the same task.
Even more so because the negotiations constructed by \Learner\ are not necessarily sound at
every stage of the learning process.
Nevertheless we show that \Learner\ is able to recover soundness just
with membership queries.
We use Mazurkiewicz traces of a special form to designate states of
the negotiation to be
learned, as well as for tests. 
Moreover, transitions cannot be just labelled by an action, but require
trace supports.
All these objects are controlled by invariants guaranteeing that
\Learner\ can always make progress. 
While conceptually more complex, the second algorithm has a similar estimate on
the number of queries as the L$^*$ algorithm.

% %\begin{theorem}
% \noindent\textbf{Theorem [Thm.~\ref{thm:learning-executions}]}
% %  For a given distributed alphabet with $|\Proc|$ processes.
%   There is an active learning algorithm for  sound deterministic negotiations,
%   using membership queries on executions  and equivalence queries returning
%   executions. It can learn a negotiation of size $s$ using
%   $O(s(s+m)(s+\log(m)))$
%   membership queries and $s$ equivalence queries, 
%   with $m$ the size of the longest counter-example.
% %\end{theorem}
% \smallskip

% Unlike in the first theorem the complexity of the second algorithm does not
% depend on the number of processes. 
% This is because a membership query about an execution can simulate $|\Proc|$
% membership queries about local paths. 
% Of course, the cost of answering a membership query about an execution depends on
% $|\Proc|$, so this is essentially just an accounting issue. 

%%% Local Variables:
%%% mode: latex
%%% TeX-master: "m"
%%% End:

\section{Basic definitions}\label{sec:basics}
%We define basic notions related to negotiations.

A \emph{(deterministic) negotiation} describes the concurrent
behavior of a set of 
processes.
At every moment each process is in some node. 
A node has a domain, namely the set of processes required to execute
one of its actions. 
If at some moment all the processes from the domain of the node are in that
node, then they choose a common action (outcome) to perform.
In deterministic negotiations, as the ones we consider here, the
outcome determines uniquely a new node for every process.

We fix a finite set of processes $\Proc$. 
A \emph{distributed alphabet} is a set of actions $\Act$ together with a
function $\dom:\Act\to 2^\Proc$ telling what is the 
(non-empty) set of processes participating in each action.
More generally, for a sequence of actions $w \in\Act^*$ we write
$\dom(w)$ for the set of processes participating in $w$, so
$\dom(w)=\cup_{|w|_a >0} \dom(a)$.

\begin{definition}
	A \emph{negotiation diagram} over a distributed alphabet $(\Act,\dom)$ is a tuple
	%\begin{equation*}
		$\Nn=\struct{\Proc,N,\ndom, \Act,\dom,$ $\delta,\ninit,\nfin}$,
	%\end{equation*}
	where
	\begin{itemize}
		\item $\Proc=\set{p,q,\dots}$ is a finite set of \emph{processes};
		\item $N=\set{m,n,\dots}$ is a finite set of
                  \emph{nodes}, each node $n$ has a non-empty domain
                  $\ndom(n) \subseteq \Proc$;
   %\item $(\Act,\dom)$ is a distributed alphabet;
		\item $\ninit$ is the initial node, $\nfin$ the
                  final one, and $\ndom(\ninit)=\ndom(\nfin)=\Proc$;
		% \item $\ndom: N\to 2^\Proc$ determines the domain of
                %   each node;
                % \item  $\dom: \Act \to 2^\Proc$ determines the domain of each action;
		\item $\d :N\times \Act\times \Proc \stackrel{.}{\to} N$ is a partial function
		defining the transitions.%dynamics of the negotiation. 
	\end{itemize}
	We also require that domains of nodes and actions match:
	\begin{itemize}
		\item if $n'=\d(n,a,p)$ is defined then
                  $\ndom(n)=\dom(a)$, $p\in\dom(a) \cap \ndom(n')$, and
		$\d(n,a,q)$ is 	defined for all $q\in
                \dom(a)$.
             	\end{itemize}
The \emph{size} of $\Nn$ is $|N|+|\de|$.
\end{definition}

A \emph{configuration} is a function $C:\Proc \to N$ indicating for
each process in which  node it is.
A node $n$ is \emph{enabled} in a configuration $C$ if all processes from the domain
of $n$ are at node $n$, namely, $C(p)=n$ for all $p\in\ndom(n)$.
Note that any two simultaneously enabled nodes $n,n'$ have disjoint domains, $\ndom(n)
\cap \ndom(n')=\es$.
We say that $a$ is an outgoing action from $n$ if $\d(n,a,p)$ is defined, denoted
$a\in\out(n)$.
% Notice that if $\d(n,a,p)$ is defined for some $p$ then $\d(n,a,q)$ is defined
% for every $q\in\ndom(n)=\dom(a)$.
If $n$ is enabled in $C$ and $a\in\out(n)$ then a transition to a
new configuration $C \act{a} C'$ is possible, where $C'(p)=\d(n,a,p)$ for
all $p\in\dom(a)$, and $C'(p)=C(p)$ for $p\not\in\dom(a)$.
As usual, we write $C \act{} C'$ when there is some $a$ with $C
\act{a} C'$, and $\act{*}$ is the reflexive-transitive closure of $\act{}$.

The \emph{initial configuration} $\Cinit$ is the one with
$\Cinit(p)=\ninit$ for all $p$. The \emph{final configuration} $\Cfin$ is such
that $\Cfin(p)=\nfin$ for all $p$. 
% So $n$ is a \emph{final node}  if $\out(n)=\es$, and $\Cc$ is a final
% configuration if $\Cc(p)$ is a final node for every $p\in\Proc$.
%There may be several final configurations.

An \emph{execution} is a sequence of transitions between configurations starting in the initial
configuration
\begin{equation*}
	\Cinit=C_1\act{a_1}C_2\act{a_2}\dots \act{a_i}C_{i+1}\ .
\end{equation*}
%A tuple $(n,a)$ is called \emph{location}. 
%We use $\ell$ to range over locations. 
Observe that an execution is determined by  a sequence of
actions.
A \emph{successful execution} is one  ending in $\Cfin$.
The language $L(\Nn)$ of a negotiation is the set of successful
executions, $L(\Nn)=\set{w \in \Act^* : \Cinit \act{w} \Cfin}$.

The \emph{graph} of $\Nn$ has the set of nodes $N$ as vertices and
edges $n\act{(a,p)} n'$ if $n'=\d(n,a,p)$.
A \emph{local path} is a path in this graph, and $\paths(\Nn)$ denotes
the set of local paths of negotiation $\Nn$, leading from the initial
node $\ninit$ to the final node $\nfin$.
W.l.o.g.~we assume that each node belongs to some local path from
$\ninit$ to $\nfin$.
In a deterministic negotiation there is at most one outgoing action for every
pair action/process $(b,p)$.
We prefer to write it as $b_p$.
For example, $(\aappl)_{NA}(\asetup)_{NA}(\adec)_{EC}$ is a local path in the
negotiation from Figure~\ref{fig:lipics}.
The alphabet of local paths is then  $A_{\dom}=\set{a_p : a \in \Act, p \in\dom(a)}$.
Clearly, $\paths(\Nn)$ is a regular language over alphabet $A_{\dom}$.
For a sequence $w \in \Act^*$ and a process $p$ we write $w|_p$ for
the \emph{projection} of $w$ on the set of actions having $p$ in their
domain.
Note that these projections are, in particular, local paths.
We often consider projections $w|_p=a_1 \dots a_k$ as words over
alphabet $A_{\dom}$,  namely $(a_1)_p \dots (a_k)_p \in A_\dom^*$. 
Coming back to Figure~\ref{fig:lipics},
the projection on $NA$ of the complete execution
$(\aappl)(\asetup)(\adinit)(\aack)(\asvote)(\avote)(\adec)$ is 
the local path $(\aappl)_{NA}(\asetup)_{NA}(\adec)_{NA}$.
\medskip

A negotiation diagram is \emph{sound} if every execution
$\Cinit \act{*} C$ can be extended to a successful one, so $\Cinit \act{*} C
\act{*} \Cfin$. 

A sound negotiation cannot have a deadlock, i.e., a
configuration that is not final but from where no process can move.
An example of a deadlock configuration is when  process $p$ is at node $n_p$ with
domain containing $\set{p,q}$,
and process $q$ is at node $n_q\not=n_p$ also with the domain containing $\set{p,q}$. 
Another possibility for a negotiation to be unsound is to have an
execution that loops without the possibility of exiting the loop.

Sound, deterministic negotiations enjoy a lot of structure, in
particular they can be decomposed hierarchically using finite automata and
partial orders~\cite{EspMusWal:17}.
A notable property we will use often is that for every node $n$ there
is a unique reachable configuration in which node $n$ is the
unique enabled node: %. This configuration is denoted $I(n)$. \anca{expand}
\begin{theorem}[Configuration $I(n)$\cite{EspMusWal:17}]\label{th:I}
  Let $\Nn$ be a sound and deterministic negotiation.
  For every node $n$ there exists unique configuration $I(n)$ such that
  node $n$ is the only node enabled in $I(n)$.
\end{theorem}
The uniqueness property from this theorem is very powerful, whenever we have an
execution $\Cinit\act{*} C$, and $n$ is the only node enabled in $C$ then we
know that $C=I(n)$, so we know where all the processes are.

\paragraph{Mazurkiewicz traces.}\label{def:traces}
For a given distributed alphabet $(\Act,\dom)$, an
equivalence relation $\tequiv$ on $\Act^*$ is defined as the transitive
closure of  $uabv \tequiv uba v$, for $\dom(a) \cap \dom(b)=\es$, $u,v
\in \Act^*$.
A \emph{Mazurkiewicz trace} is a $\tequiv$-equivalence class, and a
trace language is a language closed under $\tequiv$.
Note that languages of negotiations are trace languages.
We identify a word over $\Act$ with its
$\tequiv$-equivalence class, so the trace it represents.
Alternatively, a trace can be as a labeled
partial order of a special kind.
Finally let us introduce some notation about prefixes and suffixes of traces. 
When $w\in\Act^*$, we write $\min(w)$, for the set $\set{a \in \Act :
  w \tequiv aw' \text{ for some } w'\in Act^*}$ of minimal actions of
$w$.
Given $u,w \in \Act^*$ we say that the  $u$ is a \emph{trace-prefix} of
$w$ if there is some 
$v \in\Act^*$ such that $uv \tequiv w$.
In this case we call $v$ a trace-suffix of $w$, and we denote it by $u^{-1}w$.
% % We often omit the word ``trace'' when referring to (trace)
% % prefixes and suffixes.
% % If $u \sqs w$  we write $u^{-1}w$ to denote the trace-suffix $v$ with
% % $uv \tequiv w$.
% Finally, let $v,s$ be a  trace-prefix and trace-suffix, respectively, of $w$.
% We write $v \cap  s$ to denote the maximal trace-prefix $r \sqs s$ that is
% also a trace-suffix of $v$, so $v \tequiv v'r$ and $s \tequiv rs'$ for
% some $v',s'$.
% With $v,s$ as above we refer to $v'$ as $v \setminus s$, so $v \tequiv
% (v\setminus s) (v \cap s)$.
% Finally, let $v,r$ be a  prefix and suffix, respectively, of $w$.
% We write $v \cap  r$ to denote the maximal prefix $r' \sqs r$ that is
% also a suffix of $v$, so $v \tequiv v'r'$ and $r \tequiv r's$ for
% some $v',s$.
% With $v,r$ as above we refer to $v'$ as $v \setminus r$.

%  Observe that if $w$ is an execution in $\Nn$ from some
% configuration $C$
% then $w|_p$ is a local path from $C(p)$.

% A negotiation diagram is \emph{acyclic} if its graph is acyclic.
% So an acyclic negotiation cannot have infinite executions.

\section{Minimizing negotiations}\label{sec:minimization}

We show now a close connection between sound deterministic
negotiations and finite automata.  
An interesting consequence is that sound deterministic negotiations can be
minimized, and that the minimal negotiation is unique.

Here we will work with local paths as defined in Section~\ref{sec:basics}.
Recall that these are sequences over alphabet $A_{\dom}=\set{a_p : a \in \Act, p
\in\dom(a)}$ labelling paths in the graph of a negotiation.  
In particular a projection $w|_p$ of an execution $w$ is a local path. 
The following simple observation about projections will be useful.

% We start with a definition capturing the intuition that a local path is a part of
% an execution.

% \begin{definition}
%\paragraph{Compatible local path.}

% A \emph{projection of an execution} $w$ on process $p$ is a local path $w|_p$,
% the subsequence of $w$ consisting of actions of process $p$.
% For example, for an execution $(\aappl)(\asetup)(\adinit)(\aack)(\asvote)(\avote)(\adec)$ of
% $(\aappl)_T(\asetup)_T(\asvote)_T(\adec)_T$. 

% The are local paths that are not projections of executions because they mix
% transitions of different processes. 
% A local path $\pi \in A^*$ is \emph{compatible  with an execution} $w\in\Act^*$
% when either $\pi$ is empty, or $\pi=a_p \pi'$, $w=uav$ with $p \notin \dom(u)$
% and $\pi'$ is compatible with $v$.  
% We say that $\pi$ is \emph{maximally compatible with $w$} if 
% there is no $b_q \in A$ such that $\pi b_q$ is compatible with $w$.
% For example $(\aappl)_T(\asetup)_C(\adec)_{EC}$ is a local path maximally
% compatible with the execution above. 

% %\end{definition}

% We start with several observations on relations between executions and local
% paths. 

% \begin{lemma}\label{lem:compatible-execution-path}
% \igw{Maybe lemma not needed}Suppose $n$ is a node of a  sound deterministic negotiation $\Nn$, and $w$ is an
% execution from $I(n)$, namely, $I(n) \act{w} C$.
% If $\pi$ is compatible with $w$, then $\pi$ is a local path from $n$ in $\Nn$.
% Moreover, if $\pi$ is maximally compatible with $w$ and the last action of $\pi$
% is $a_p$, then $\pi$ ends in  node $C(p)$.  
% \end{lemma}

\begin{lemma}\label{lem:locpath}
  Let $\Nn$  be a deterministic negotiation, $C\act{u} C'$ an
  execution in $\Nn$, and $p$ a process. The projection $u|_p$ of $u$ on $p$ is a local
  path in $\Nn$ from $C(p)$ to $C'(p)$. 
\end{lemma}

The automata we will consider in the paper are deterministic (DFA), but
incomplete.
A DFA $\Aa$ will be written as
$\Aa=\struct{S,A,\out,\d,s^0,F}$, with $S$ as a set of
states, $\d:S \times A \to S$ a partial function, and $\out:S \to
2^A$ a map from states to their set of
outgoing actions.
Thus, $a \in \out(s)$ iff $\d(s,a)$ is defined.
While $\out$ seems redundant, it is very convenient when learning
incomplete automata, as we do in this paper. 
The next definition states a useful property of automata accepting $\paths(\Nn)$. 

%\paragraph{Dom-complete automata.}
% The next definition singles out a property of automata recognizing local paths
% of negotiations. 
 \begin{definition}[Dom-complete automata]\label{def:domcomplete}
  A finite automaton $\Aa$ over the alphabet $A_{\dom}$ is \emph{dom-complete} if for every
  state $s$ of $\Aa$  and every $a_p,a_q,b_q \in A_\dom$:
  \begin{enumerate}
    \item $a_p \in out(s)$ iff $a_q \in \out(s)$, and
    \item if $\set{a_p,b_q}  \incl \out(s)$ then $\dom(a)=\dom(b)$.
  \end{enumerate}
    Moreover, we require that $a_p \in \out(s_\init)$ for some $a$
    with $\dom(a)=\Proc$, where $s_\init$ is the initial state of $\Aa$.
\end{definition}

\begin{remark}\label{rem:synt-path}
  Observe that every trimmed DFA $\Aa$ accepting the language $\paths(\Nn)$ for
  $\Nn$  sound and deterministic, is dom-complete, if $\Nn$ has at least one transition.
  (An automaton is \emph{trimmed} if every state is reachable from the initial
  state and co-reachable from some final state).
  To see this consider a state $s$ of $\Aa$.
  As $\Aa$ is trimmed, there is some $\pi \in A_{\dom}^*$ with $s_0\act{\pi} s$.
  Consider  $\set{a_p,b_q} \subseteq \out(s)$.
  Once again thanks to trimness,   $\pi a_p$ and $\pi b_q$ are prefixes of
  some words in $\paths(\Nn)$.
  Since $\Nn$ is deterministic, $\pi$ induces a local path in $\Nn$, from $\ninit$
  to some node $n$.
  Hence, $\dom(a)=\dom(b)=\ndom(n)$ by the definition of negotiation. 
  The first property follows by a similar argument. 
\end{remark}

Let us spell out how to construct a negotiation from a dom-complete automaton. 
The conditions on the automaton are precisely those that make the result be a
negotiation. 

\begin{definition}\label{def:neg-from-automaton}
  Let $\Aa=\struct{S,A_{\dom},\out,\d_\Aa,s^0,s_f}$ be a
  dom-complete DFA %, trimmed DFA
  such that $\out(s_f)=\es$ for the unique final state $s_f$.
  We  associate with $\Aa$ the negotiation 
  %\begin{equation*}
    $\Nn_\Aa=\struct{\Proc,N,\ndom,$ $\Act,\dom,\d,\ninit,\nfin}$
  %\end{equation*}
  where
  \begin{itemize}
    \item $N=S$,  $\ninit=s^0$, and $\nfin=s_f$,
    \item % The domain of node $s$ is determined by the outgoing actions
      % in $\Aa$:
    $\ndom(s)=\dom(a)$ if $a_p \in \out (s)$ for some $a\in\Act$ and
    $p\in \Proc$; moreover, $\ndom(s_f)=\Proc$,
    \item $\d(s,a,p)=\d_\Aa(s,a_p)$ for all $s,a,p$.
  \end{itemize}
\end{definition}

The main result of this section says that the \emph{minimal} automaton of $\paths(\Nn)$ determines
a sound deterministic negotiation. 
\begin{proposition}\label{prop:mindfa-negotiation}
  Let $\Nn$ be a sound deterministic
  negotiation and $\Aa$ the minimal DFA accepting
  $\paths(\Nn)$.
  Then $L(\Nn)=L(\Nn_\Aa)$. Moreover $\Nn_\Aa$ is deterministic and sound.
\end{proposition}

\begin{corollary}\label{cor:projections} 
  Let $\Nn$ be sound and deterministic, and let $w \in\Act^*$ be such that
  $w|_p \in \paths(\Nn)$ for all $p \in \Proc$. Then $w \in L(\Nn)$.
\end{corollary}

Recall that for a regular language $L$ any automaton accepting $L$ can 
be mapped homomorphically to the minimal automaton of $L$. For
deterministic, sound negotiations we have the same phenomenon, where
homomorphisms map nodes to nodes, so that transitions are mapped to
transitions with the same label.

 \begin{corollary}\label{cor:minimal-negotiation}
   Let $\Nn$ be sound, deterministic, and let $\Aa$ be the minimal DFA accepting
   $\paths(\Nn)$. Then there is a homomorphism from $\Nn$ to $\Nn_\Aa$.
 \end{corollary}

\begin{corollary}\label{cor:equiv}\anca{added}
  Language equivalence of sound, deterministic negotiations can be
  checked in \PTIME.
\end{corollary}
%%% Local Variables:
%%% mode: latex
%%% TeX-master: "m"
%%% End:

\section{Angluin learning  for finite automata}\label{sec:angluin}

We briefly present a variant of  Angluin's $L^*$ learning algorithm for finite automata.
Our approach is particular because it works with automata that are not
necessarily complete. 
This will be very useful when we extend the algorithm to learn negotiations.
The automata coming from negotiations, dom-complete automata as in
Definition~\ref{def:domcomplete}, are in general not complete.

Angluin-style learning of finite automata relies on the Myhill-Nerode equivalence
relation, which in turn accounts for the unicity of the minimal DFA of
a regular language.
A \Learner\ wants to compute the minimal DFA $\Aa$ of an unknown regular
language $L \subseteq A^*$.
For this she interacts with a \Teacher\ by
asking membership queries $w \in^? L$ and equivalence
queries $L(\tAa) =^? L$, for some word $w$ or automaton $\tAa$.
To the first type of query Teacher replies yes or no, to the second Teacher
either says yes, or provides a word that is a counterexample to the equality of the
two languages.

Angluin's algorithm maintains two finite sets of words, a set $Q
\subseteq  A^*$
of \emph{state words} and a set $T \subseteq A^*$ of \emph{test words}.
% In our variant we  also record the counter-examples given by
% \emph{Teacher} as two sets of words:  $\Pplus\incl L$ and $\Pminus\cap L=\es$.
The sets $Q,T$ are used to construct a deterministic candidate
automaton $\tAa$ for $L$.
%For convenience we will denote $\tAa$ as $\tAa$, and $\tL=L(\tAa)$.
The elements of $Q$ are the states of $\tAa$.
The set $Q$ is prefix-closed and $\e \in Q$ is the initial state of $\tAa$.
% The sets $\Pplus$ and $\Pminus$ provide some checks on $\tAa$ that allow to enlarge
% it if a check does not succeed.

The set of words $T \subseteq A^*$ determines an equivalence
relation $\eqT$ on $A^*$ approximating Myhill-Nerode's right congruence
$\equiv^L$ of $L$: 
%\begin{equation*}
  $u\eqT v \quad\text{if}\quad \text{for all $t\in T$, \; $ut\in L$ iff $vt\in L$}$.
%\end{equation*}
Angluin's algorithm maintains two invariants, \Unique\ and \Closure.

\begin{description}
 \item[\Unique] \text{for all $u,v\in Q$, if $u\eqT v$ then
     $u=v$}\,.
\end{description}

Observe that if  $u\equiv^L v$ then $u\eqT v$. 
So $\eqT$ has no more equivalence classes than the Myhill-Nerode's congruence
$\equiv^L$.
Since Angluin's algorithm adds at least one state in every round, the
consequence of \Unique\ is that the number of rounds is bounded by the index of $\equiv^L$,
or equivalently by the size of the minimal automaton for $L$.

% A notable fact is that $\eqT$ may not be a right congruence. 
% So the intermediate $\tAa$ constructed by Learner may have a very different structure
% than the final automaton $\Aa$.
% For example, it may happen that two $\eqL$-equivalent words end up in different
% states of $\tAa$.

In the original Angluin's algorithm the candidate automata maintained by
\Learner\ are complete, every state has an outgoing transition on every
letter.
% We drop this assumption, and add only transitions that are required by
% positive counter-examples.
% It is not clear what is the good choice from a practical perspective. 
% There are examples in the literature where alphabets are very large~\cite{vaa17},
%  but the maximal out-degree of nodes is rather small. 
% % Alphabets used by negotiations obtained from workflow nets
% % are large in general, but the maximal out-degree of nodes may be rather
% % small. 
% Anyway, it is straightforward to complete an automaton if in some context this choice turns
% out to be more efficient. 
When learning negotiations, it is more natural to work with automata that
are incomplete.
Because of this we have a third parameter besides
$Q,T$, which is a mapping  $\out:Q\to 2^A$, telling for each state
what are its outgoing transitions defined so far.
The original closure condition of Angluin's algorithm
now becomes: 
\begin{description}
 \item[\Closure] \; \text{for all $u \in Q, a \in \out(u)$
     there exists } $v \in Q$ \text{ with }  $ua \eqT v$\ .
\end{description}

For $(Q,T,\out)$ satisfying \Unique\ and \Closure\ we can now construct an automaton:
%\begin{equation*}
  $\tAa=\struct{Q,A,\de,q_\init,F}$
%\end{equation*}
with  state space $Q$ and alphabet $A$.
The initial state is $\e$, and the final states of $\tAa$ are the states $u \in
Q \cap L$.
The partial transition function  $\de:Q \times A \stackrel{.}{\to} Q$ is
defined by:
\begin{equation*}
  \d(u,a)=v\quad \text{ if }\quad  \text{$u a\eqT v$ and $a\in\out(u)$}\ .
\end{equation*}
Thanks to \Unique\ there can be at most one  $v$ as above.
While \Closure\  guarantees that $\d(u,a)$ is defined iff $a\in\out(u)$.

The learning algorithm works as follows.
Initially, $Q=T=\set{\e}$ and $\out(\e)=\es$.
Note that $(Q,T,\out)$ satisfies \Unique\ and \Closure.
The algorithm proceeds in rounds.
A round starts with $(Q,T,\out)$ satisfying both invariants. 
\Learner\ can construct a candidate automaton $\tAa$.
She then asks \Teacher\ if $\tAa$ and $\Aa$ are equivalent. 
If yes, the algorithm stops, otherwise \Teacher\ provides a counter-example word
$w \in A^*$.
It may be a positive counter-example, $w\in L\setminus L(\tAa)$, or a negative
one, $w\in L(\tAa)\setminus L$.
%One of the two cases described below can happen. 
In both cases \Learner\ extends $(Q,T,\out)$ while preserving
the invariants.
Then a new round can start.
The details can be found in the Appendix.

\section{Learning negotiations with local queries}\label{sec:paths}

% \section{Learning with local path membership queries}
% \section{Learning with execution membership queries}
% \section{Learning with local/global  membership queries}

We present our first algorithm for learning sound deterministic negotiations.
% Given Proposition~\ref{prop:mindfa-negotiation} we could use Angluin learning
% for finite automata to learn negotiations. 
% This supposes that Learner can ask questions about local paths, and Teacher
% replies with local paths as counter-examples. 
% Here we consider the scenario when Teacher replies with a complete execution,
% not with a local path.
% This is more realistic if we assume that  Teacher does not have access to
% the internals of the negotiation.  
This algorithm serves as intermediate step to the main learning
algorithm of Section~\ref{sec:executions} that uses
only executions as queries.

Recall that an execution is a sequence over $\Act$; where $\Act$ is an alphabet
of actions equipped with a domain function $\dom:\Act\to2^\Proc$.
Local paths are  sequences over the alphabet $A_{\dom}=\set{a_p : a \in \Act, p \in
\dom(a)}$.
They correspond to paths in the graph of the negotiation. 

% For $\Nn$ a deterministic negotiation and a node $n$ we denote by
% $\paths(n) \subseteq A^*$  the set of local paths
% from node $n$ to the final node $\nfin$.
% We also write $\paths(\Nn)$ instead of $\paths(\ninit)$.

% Our first algorithm uses two kinds of queries: membership queries
% $\pi\stackrel{?}{\in} \paths(\Nn)$ for $\pi \in A^*$, and equivalence
% queries between the candidate 
% negotiation and $\Nn$. 
% Counter-examples provided by \Teacher\ are executions, so
% sequences from $\Act^*$.
%The main challenge is how to use these counter-examples to discover new states. 

% We fix a set of processes $\Proc$ and a distributed alphabet
% $(\Act,\dom:\Act\to2^\Proc)$.
We assume that \Teacher\ knows a sound deterministic  negotiation $\Nn$
over the distributed alphabet $(\Act,\dom:\Act\to2^\Proc)$. %, and set $L=L(\Nn)$. 
\Learner\ wants to determine the minimal negotiation $\tNn$ with $L(\tNn)=L$. 
By Corollary~\ref{cor:minimal-negotiation} this minimal negotiation is
$\Nn_\Aa$, with $\Aa$ the minimal
automaton for the regular language $\paths(\Nn)$.
Our  algorithm uses two types of queries:
% Iteratively \Learner\ constructs $\tNn$ using
% information obtained from two types of queries:
\begin{itemize}
  \item membership queries $\pi \in^? \paths(\Nn)$, to which \Teacher\
  replies yes or no; 
  \item equivalence queries: $L(\tNn) =^?L(\Nn)$ to which \Teacher\ either replies
  yes, or gives an execution $w \in \Act^*$ in the symmetric difference of $L(\tNn)$ and
  $L(\Nn)$. 
\end{itemize}
The structure of the algorithm will be very similar to the one for DFA from Section~\ref{sec:angluin}.
Let us explain two new issues we need to deal with.
% Using Corollary~\ref{cor:minimal-negotiation} we aim at learning the minimal
% automaton for the regular language $\paths(\Nn)$.
Learner will keep a tuple  $(Q,T,\out)$, with $Q,T \subseteq
\Adom^*$ and $\out:Q \to  2^{\Adom}$,  satisfying invariants
\Unique\ and \Closure.
This tuple defines an automaton $\tAa$ as in Section~\ref{sec:angluin}.
Learner constructs from $\tAa$ a negotiation $\tNn$ as in
Definition~\ref{def:neg-from-automaton}.  
She proposes $\tNn$ to Teacher, and if Teacher answers with a
counter-example execution she uses it to extend $(Q,T,\out)$ and construct a new $\tNn$.
Compared to learning finite automata, we have two new issues.
We need to impose additional invariants to obtain a dom-complete automaton
$\tAa$ (Definition~\ref{def:domcomplete}) as this is required to construct $\tNn$.
More importantly, we need to find a way how to exploit a counter-example that is
an execution and not a local path (Lemmas~\ref{lem:lp-pos} and~\ref{lem:lp-neg}).

We will write $L_P$ as shorthand for $\paths(\Nn)$. 
Since final nodes of negotiations do not have outgoing actions,
$L_P$ is prefix-free.
Said differently, all words in $L_P$ are $\eqT$ equivalent as soon as
$\e\in T$.
In particular, there will be a unique  final state (with no outgoing
transitions)  in the automaton $\tAa$ constructed from $(Q,T,\out)$.
We write $[u]_T$ for the $\eqT$-class of $u \in \Adom^*$. 
The learning algorithm will preserve the  following invariants for
a triple $(Q,T,\out)$: % with $Q, T\subseteq A^*$, and $\out: Q \to 2^A$:
\begin{description}
 \item[\Unique] For all $u,v\in Q$, $u \eqT v$ implies $u=v$.
\item[\Closure] For every $u \in Q, a\in\out(u)$ there exists $v \in Q$
  with $ua \eqT v$.
 \item[\PREF] For every $u \in Q$ there is some $t \in T$ with $u t \in
   L_P$. % Equivalently, $[u]_T \not=\bot$ for all $u \in Q$.
 \item[\DOM] For every $u \in Q$ and every $a_p,a_q\in \Adom$:
 %  \begin{enumerate}
 %\item
 $a_p \in out(s)$ iff  $a_q \in \out(s)$.
%         \item if $\set{a_p,b_q}  \incl \out(s)$ then $\dom(a)=\dom(b)$.
%      \end{enumerate}
\end{description}
The first two invariants are the same as in Section~\ref{sec:angluin}.
The third one is important to determine the domain of a node:
if $u$ can be extended to a complete path, we know one outgoing action from $u$,
and this determines the domain of $u$.
The last invariant is the first condition of dom-completeness
(Definition~\ref{def:domcomplete}).
Note that \Closure\ and \PREF\ entail the other condition of
dom-completeness: if $\set{a_p,b_q} \subseteq \out(u)$ for $u \in Q$
then $ua_p$ and $ub_q$ are local paths because of \Closure\ and \PREF; 
so $u$ leads in $\Nn$ to some node $n$ with  outgoing actions  $a,b$, hence $\dom(a)=\dom(b)=\ndom(n)$.
% Additionally to the three invariants we will use the closure property
% for $(Q,T,\out)$ defined on page~\pageref{p:closure}. 

% In addition to the three invariants above we need $(Q,T,\out)$ to be closed
% (Definition~\ref{def:closure}).
%This allows to construct an automaton and a negotiation.
\begin{lemma}
  If $(Q,T,\out)$ satisfies all four invariants \Unique, \Closure,
  \PREF, \DOM\ then
  the associated automaton
  $\tAa$ is dom-complete, so a deterministic negotiation $\tNn=\Nn_{\tAa}$ can be
  defined, see~Definition~\ref{def:neg-from-automaton}.
\end{lemma}

Henceforth we use $\tNn$ to denote the negotiation $\Nn_{\tAa}$ and $\tL$ to denote the language of $\tNn$.

The next two lemmas lay the ground to handle  counter-examples
provided by Teacher. 
Suppose $(Q,T,\out)$  satisfies all four invariants.
Teacher replies with $w$ in the symmetric difference of $L$ and $\tL$.
As  $w$ is an execution, and not a local path, it 
can be seen as a (Mazurkiewicz) trace.
We will use operations on traces 
introduced on page~\pageref{def:traces}. 

The main point of the next lemma is not stated there explicitly.
An execution in a negotiation $\tNn$ may reach a deadlock. 
The lemma says that we do not have to deal with this situation because
we can look backwards either for  a place
where we need to add a  node (\NEQ) or a transition (\OUTINC).

%\paragraph{Positive example.}
%Let $w \in L \setminus \tL$ be a positive example.

\begin{lemma}\label{lem:lp-pos}
  Consider a positive counter-example $w\in L\setminus \tL$. 
  Let $v$ be the maximal trace-prefix of $w$ executable in $\tNn$. 
  So we have $\tCinit\act{v}\tC$ in $\tNn$, and no action in
  $\min(v^{-1}w)$ can   be executed from $\tC$. 
  With at most $|\Proc|$  membership queries \Learner\ can determine one of the following situations:
  \begin{description}
    \item[\OUTINC:]  An action $b\in\min(v^{-1}w)$, a node
      $u \in \Adom^*$ of $\tNn$, and a sequence $r
      \in \Act^*$  starting with $b$ such that for every $p\in\dom(b)$:
      \[ u \,
        r|_p\in L_P \quad \text{and } \quad  b_p\not\in \out(u)\,.
        \]
    \item[\NEQ:] A process $p$, and a local path $\pi \in \Adom^*$
      such that
      \[ v|_p\, \pi\in L_P \quad \not\iff \quad  u\, \pi\in L_P \quad \text{ with
        } u=\tC(p)\,.
      \]
      \end{description}
\end{lemma}

The case of negative counter-examples is much simpler, and 
we get the \NEQ\ case as in Lemma~\ref{lem:lp-pos} for $\pi=\e$:

\begin{lemma}\label{lem:lp-neg}
  Consider a negative  counter-example $w\in \tL\setminus L$, and let $\tCinit \act{w} \tC$.
  With at most $|\Proc|$  membership queries \Learner\ can find a process $p$
      such that
      \[ w|_p \in L_P \quad \not\iff \quad  u\ \in L_P \quad \text{ for
        } u=\tC(p)\,.
      \]
    \end{lemma}

% \begin{proof}
%   Since $w\in \tL$, all nodes in configuration $\tC(p)$ are accepting.
%   By definition of $\tNn$,  for every process $p$, the node $u=\tC(p)$ is such
%   that  $u\in L_P$.
%   On the other hand, by Corollary~\ref{cor:projections} there is $p$ such that
%   $w|_p \notin L_P$. 
%   Learner can find this $p$ with at most $|\Proc|$ membership queries.
%   We get $w|_p\not\in L_P$ and $u\in L_P$.
% \end{proof}

\paragraph{Processing counter-examples.}

We describe now how to deal with the two cases
\OUTINC\ and  \NEQ\ of Lemmas~\ref{lem:lp-pos} and \ref{lem:lp-neg}.
Before we start we observe that \PREF\ and \Closure\ entail the
following variant of \PREF, that will
be useful below:
\begin{description}
  \item[\PREF'] for every $u'\in Q$ and every $a_p\in\out(u')$
  there is some $t\in T$ such that $u'a_pt\in L_p$.
\end{description}
Indeed, using \Closure\ we get some $v\in Q$ with $u'a_p\eqT v$, and because of
\PREF, there is some $t$ with $vt\in L_P$.
So $u'a_pt\in L_P$.

  % \begin{description}
  %   \item[\OUTINC:]  An action $b\in\min(v^{-1}w)$, a node
  %     $u \in A^*$ of $\tNn$, and a sequence $r
  %     \in b \, \Act^*$ such that for every $p\in\dom(b)$ we have  $u \,
  %     r|_p\in L_P$ and $b_p\not\in \out(u)$.

  %   \item[\NEQ:] A process $p$, and a local path $\pi\in A^*$
  %     such that
  %     \[ v|_p\, \pi\in L_P \quad \not\iff \quad  u\, \pi\in L_P \quad \text{ for
  %       } u=\tC(p)\,,
  %     \]
  %       where $\tCinit \act{v} \tC$.
  %     \end{description}

\paragraph{\textbf{\OUTINC\ case.}}
We have some node $u\in Q$ with $ur|_p\in L_P$, for $r \in\Act^*$ starting
with $b$,  and $b_p\not\in\out(u)$ for all $p\in\dom(b)$.
% If $\out(u)\not=\es$ then consider arbitrary $a_q\in\out(u)$.
% By completeness of $(Q,T,\out)$, there is $v_{a,q}$ with
% $ua_q\eqT v_{a,q}$. 
% Then \PREF\ invariant gives us $t_{a,q}\in T$ such that $ua_qt_{a,q}\in L_P$.
% We will need these $t_{a,q}$ to finalize the argument. 

For every $p \in\dom(b)$, we add $b_p$ to $\out(u)$ and $(b^{-1}r)|_p$ to $T$.
For each $b_p$, one at a time, we check if there is some $v \in Q$
with $u b_p \eqT v$. 
%This requires $|\Proc||T|$ membership queries.
If not, we add $u b_p$ to $Q$ with $\out(u b_p)=\es$.
This step preserves \Unique\ and \DOM.
Also \PREF\ holds if $ub_p$ is added
to $Q$, because of $(b^{-1}r)|_p\in T$.
% We need to check if invariant \DOM\ holds after extending $\out(u)$. 
% Since we added $b_p$ for all $p\in\dom(b)$, the first condition holds
% for $u, b_p,b_q$.
% The second invariant clearly holds if there is no $a_q\in\out(u)$ for $a\not=b$
% and some $q$.
% Otherwise by \PREF'  we have $ua_qt\in L_P$ for some $t$. 
% Consider a node $n$ reached on $u$ in $\Nn$. 
% We have that both $a$ and $b$ are in $\out(n)$, hence $\dom(a)=\dom(b)$ by the
% definition of a negotiation. 

Finally, since $T$ changed, \Closure\ must be restored.
\Closure\ holds for newly added $ub_p$, since we set $\out(ub_p)=\es$.
The other $u'\in Q$ are those that were there already at the beginning of the
round. 
If $u'\not=u$ then $\out(u')$ is unchanged, so \PREF' continues to hold. 
For $u'=u$ we have established \PREF' by adding $(b^{-1}r)|_p$ to $T$. 
In both cases, if for some  $a_p \in \out(u')$ there is no $v\in Q$
with $u' a_p \eqT v$ then we add $u'a_p$ to $Q$, and set $\out(ua_p)=\es$.
Thanks to \PREF', invariant \PREF\ holds after this extension. 
The other invariants are clearly preserved.
% $|\Proc||\out|$
% The membership queries required to re-establish \Closure\
% for $u'\not=u$ and $a_p$ 
% because it suffices to check the $v \in Q$ that witnessed \Closure\
% for $u',a_p$ before
% enlarging $T$ by at most $|\Proc|$ test words.
% The same applies for $u'=u$ and $b_p$ such that we did
% not need to add $ub_p$ to $Q$ at the first step.
% % So re-establishing \Closure\ in one round requires $O(|\Proc|(|T|+|\out|))$ membership
% % queries.
% \anca{double-check} 

\paragraph{\textbf{\NEQ\ case.}}
We have a process $p$, a node $u\in Q$, a sequence $v\in \Act^*$,  and a local path $\pi$
such that $v|_p\, \pi\in L_P\not\iff u\pi \in  L_P$.
Moreover $\tCinit\act{v}\tC$ and $\tC(p)=u$.

Let $v|_p=a_1 \dots a_k$ and $\e \act{a_1} u_1
\dots \act{a_k} u_k$ the run of $\tAa$ on $v|_p$ (this run exists since $\tCinit \act{v} \tC$).
We have $u=u_k$ and $a_1 \dots a_k \, \pi \in L_P \not\iff u_k \, \pi
\in L_P$.
So there is some $i \in\set{1,\dots,k-1}$ such that
\[
u_i \, a_{i+1} \dots a_k \pi \in L_P \not\iff u_{i+1} \, a_{i+2} \dots
a_k \pi \in L_P
  \]
Such an $i$ can be determined by binary search using $O(\log(k))$
membership queries.
We add $a_{i+2} \dots a_k \pi$ to $T$, $u_i a_{i+1}$ to $Q$, and set
$\out(u_i a_{i+1})=\es$.
The invariants \Unique\ and \DOM\ are clearly preserved.
For \PREF\ note that $a_{i+1}$ already belonged to $\out(u_i)$, so
thanks to \PREF', invariant \PREF\ holds for $u_ia_{i+1}$ as well. 
For \Closure\ we proceed as in case \OUTINC, by enlarging $Q$,  if
necessary.

% The number of membership queries is at most $|\out|$ for
% re-establishing \Closure, and $\log(k)$ for the binary search, with
% $k$ the maximal length of counter-examples.

\paragraph{Learning algorithm.}
We sum up the developments in  this section in the learning algorithm
shown below.
The initialization step of our algorithm consists in asking
\Teacher\ an equivalence query for the empty negotiation $\Nn_\es$; this is a
negotiation consisting of two nodes $\ninit,\nfin$ and empty transition mapping $\de$.
\Teacher\ either says yes or returns a positive example $w \in L$.
Note that the first action of $w$ must involve all the processes because the domain
of the initial node is the set of all processes.
So $w=bw'$ for some $b \in\Act$ with $\dom(b)=\Proc$.
We initialize $(Q,T,\out)$ by setting $Q=\set{\e}$, $T=\set{w|_p : p
  \in \Proc}$,  and $\out(\e)=\es$.
All invariants are clearly satisfied.
Observe that we have here the \OUTINC\ case of Lemma~\ref{lem:lp-pos}
with $b \in \Act$ as above, $u=v=\e$ and $r=w$.

\begin{algorithm}%[H]
        %\KwIn{}
        %\KwOut{}
 %\TitleOfAlgo{Learning sound negotiations by querying local paths} 
   $(\res,w) \larr \EQ(\Nn_\es$)\;
   \lIf{($\res=\true$)}{\Return $\Nn_\es$}
   $(Q,T,\out) \larr (\set{\e}, \set{w|_p : p \in\Proc},
   \out(\e)=\es)$ 
   $\OUT(w,Q,T,\out)$ \tcp*{add missing transitions}
   $\CLOS(Q,T,\out)$ \tcp*{restore \Closure}
        %\BlankLine
   \While{($\res=\false$)}{
     $\tNn\larr \Negotiation(Q,T,\out)$\tcp*{build $\tNn$}
     $(\res,w) \larr \EQ(\tNn)$ \tcp*{ask \Teacher}
     \lIf{($\res=\true$)}{\Return $\tNn$}
     \lIf(\tcp*[f]{add missing transitions}){(\OUTINC)}{$\OUT(w,Q,T,\out)$}
     \lIf(\tcp*[f]{add new state}){(\NEQ)}{$\BinS(w,Q,T,\out)$}
     $\CLOS(Q,T,\out)$ \tcp*{restore \Closure}
     }\caption{Learning sound negotiations with membership queries about local paths.}
\end{algorithm}

Procedure $\OUT(w,Q,T,\out)$ adds missing transitions as described in case
\OUTINC.
It extends $\out$, and possibly $T$. 
%For efficiency it can add all the transitions and not just one. 
After calling $\OUT$ the invariants \Unique, \PREF, \DOM\
are satisfied.
Each $\OUT$ is followed by $\CLOS$ that restores the  \Closure\ invariant,
as also described in case \OUTINC. 
It may happen that nothing is added by $\CLOS$ operation.
Procedure $\BinS$ performs a binary search and extends $Q,T,\out$ as
described in the \NEQ\ case. 
After its call we are sure that \Closure\ does not hold, so $\CLOS$ adds
at least one new node. 
Thus in every iteration the algorithm extends at least one of $\out$
or $Q$.
For the complexity of Algorithm~1, see the appendix.

% The number of equivalence queries is equal to the number of iterations of the
% loop.
% By the above, it is bounded by the size of the negotiation (that is the sum of
% the number of nodes and the number of transitions). 
% Note also that $|T| \le |Q|+|\out|$, since $\OUT$ adds one element to $T$
% for each new transition, respectively~$\BinS$ adds one element per call.
% Let us estimate the number of membership queries.
% The calls of $\OUT$ altogether make $O(|\out||T|)$ membership queries,
% We can over-approximate this by $O(|\out|^2)$.
% The same complexity  holds for the calls of $\CLOS$ because this
% procedure checks  $u',a_p,v$ with $u'a_p \eqT v$ before
% enlarging $T$, only w.r.t.~newly added words in $T$.
% Finally, checking whether case \OUTINC\ or \NEQ\ holds accounts for
% $|Q| \cdot |\Proc|$ membership queries.

\begin{theorem}\label{thm:learning-paths}
  Algorithm~1 actively learns sound deterministic negotiations,
  using membership queries on local paths and equivalence queries returning
  executions. It can learn a negotiation of size $s$ using $O(s(s+
  |\Proc|+\log(m)))$ membership queries and $s$ equivalence queries,
  with $m$ the size of the longest counter-example.
\end{theorem}

It is possible to modify Algorithm~1 so that equivalence queries are
asked only for $\tNn$ sound.
We do this in our second, main Algorithm~2.
Here the presentation is clearer without this step.

%%% Local Variables:
%%% mode: latex
%%% TeX-master: "m"
%%% End:
  
%\eject
\section{Learning negotiations by querying executions}\label{sec:executions}
Our second learning algorithm asks membership queries about executions and not
about local paths. 
The immediate consequence is that $Q$ and $T$ are built from executions and not
from local paths.
Executions are sequences of actions from $\Act$, but since $\Act$ is a
distributed alphabet we consider them as (Mazurkiewicz) traces.
The trace structure of executions will be essential. 
The challenge is how to construct membership queries about executions, and how
to extract useful information from the answers. 
% Membership queries about local paths allowed Learner to ask questions about nodes in the graph of
% a negotiation. 
% It is not evident how to use executions to accomplish the same task.

% We will use the operations from page~\ref{def:traces}.
% More precisely, Learner will keep in $T$ co-prime (Mazurkiewicz) traces, and in $Q$ left contexts of
% co-prime  traces.
% In addition we will need to keep also supports for transitions, which
% will be again co-prime  traces, acting as contexts when a process
% moves in the negotiation from a node to the next one.

Throughout the section we fix the sound deterministic negotiation
$\Nn$ we want to learn.
We use the same notations as in Section~\ref{sec:paths}, namely $L,
\tNn, \tL$.
% In particular, we also fix the distributed alphabet $(\Act,\dom)$.
% For short we will write $L$ instead of $L(\Nn)$ for the language of $\Nn$.
% We write $\tNn$ for the negotiation being constructed in the learning process. 
% We write $\tL$ instead of $L(\tNn)$. 
The negotiation $\tNn$ will always be deterministic, but not necessarily sound.
Yet, we will show that \Learner\ can extend it to a sound negotiation with just
membership queries. 
So $\tNn$ will be sound at every equivalence query. 
This greatly simplifies dealing with counter-examples. 

The construction is spread over several subsections. 
First, we describe how we use Mazurkiewicz traces to identify nodes in a negotiation (Figure~\ref{fig:ut}).
Building on this we can identify transitions in negotiations. 
In Section~\ref{subsec:learning-structure} we describe our representation of nodes and
transitions of a negotiation in a learning algorithm. 
We also state there the invariants of the construction. 
Section~\ref{sec:two-operations} describes two operations for extending $\tNn$.
They are used in Sections~\ref{sec:negative} and~\ref{sec:positive} where we
show how to handle counter-examples.
Section~\ref{sec:making-sound} explains how to restore soundness of $\tNn$.
Finally, we present a learning algorithm in
Section~\ref{sec:learning-executions-algorithm}. 

\subsection{Technical set-up}

We describe how to use traces to talk about nodes and transitions in a
negotiation. 
We start with a couple of definitions. 

We use $u,v,w,s,r,t\in\Act^*$ for sequences of actions and
often consider them as partial orders, i.e., as Mazurkiewicz
traces.
Recall that we write $u \tequiv v$ when $u,v$ represent the
same Mazurkiewicz  trace.
For all other notations related to traces and configurations we refer to the end of
Section~\ref{sec:basics}.
We will use extensively Theorem~\ref{th:I} stating the existence and
uniqueness of the configuration $I(n)$ enabling precisely node $n$. 
%In $I(n)$ node $n$ is the unique enabled node.

We start by defining two main kinds of traces used throughout
the section (see~Figure~\ref{fig:ut}).
\begin{itemize}
	\item $t \in Act^*$ is \emph{co-prime} if $t$ has a unique
          minimal element in the trace order.
          In other words, there is some $b \in \Act$ such that 
          every $v \in Act^*$ with $t \tequiv v$ starts with $b$.
          We write $b=\min(t)$ and $\dmin(t)$ for the domain of
          $\min(t)$, namely, $\dmin(t)=\dom(b)$. 
	\item $s \in Act^*$ is a \emph{$(b,p)$-step} if $p\in\dom(b)$,
          $s=bs'$ is co-prime and $b$ is 
	the only action involving $p$ in $s$, namely,	$p\not\in\dom(s')$.
      \end{itemize}

\begin{figure} 
  \centering
  \includegraphics[scale=.25]{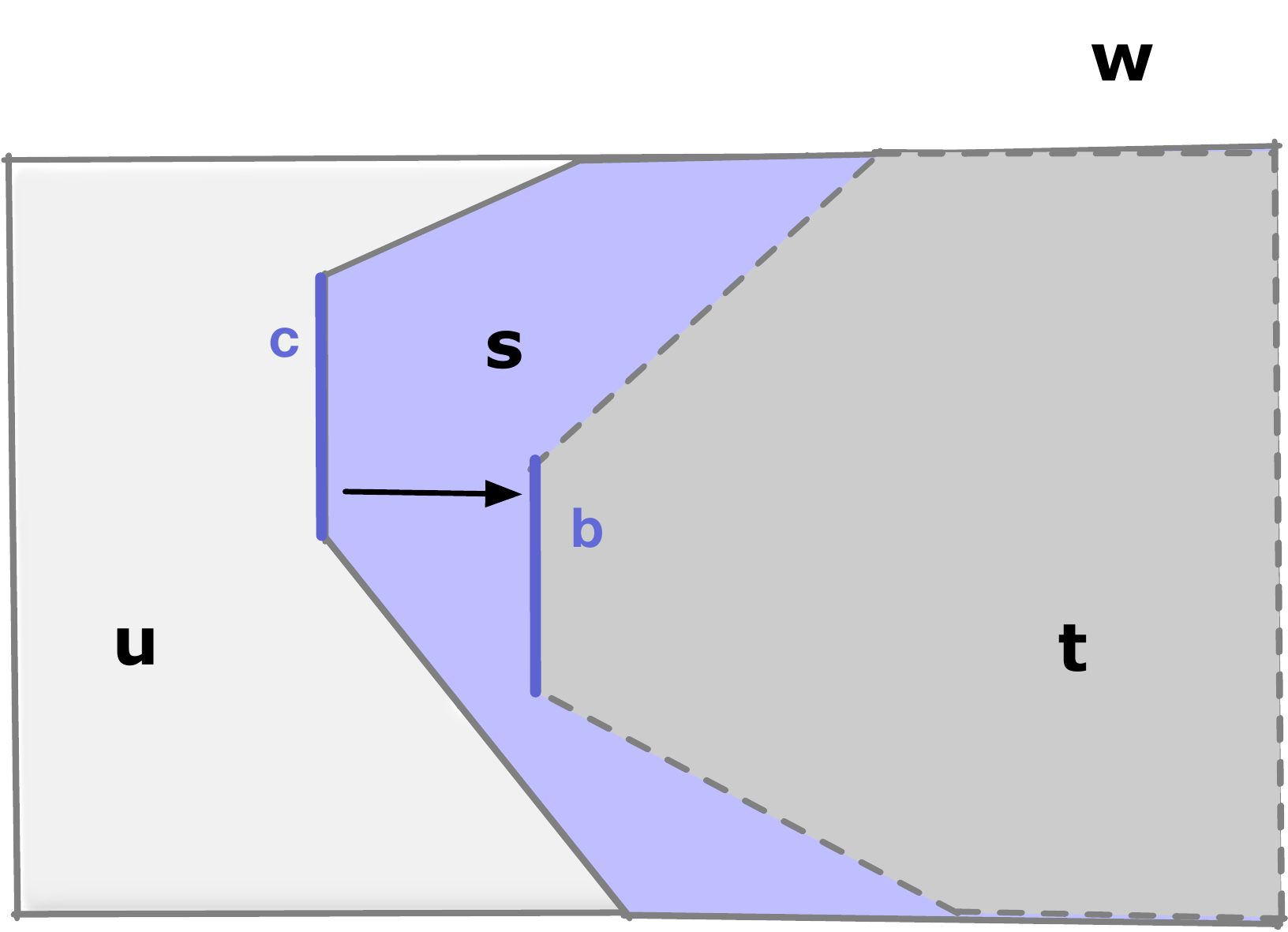}
  \caption{% Left: $u \in Q$, co-prime $t \in T$. Together $u$ and $t$ form an
  % execution. Trace $u$ designates one node in $\tNn$.
  Partial order of execution $ust$. The blue
  part $s$ is a $(c,q)$-step,
    of some process $q \in \dom(b)\cap\dom(c)$. No action of  $q$,
    besides $c$, appears in $s$. Both $t$ and $st$ are co-prime. Actions $c,b$ are
    outcomes of two nodes, and process $q$ participates in both.}
	\Description{none}
	\label{fig:ut}
\end{figure}

The next two lemmas explain the link between co-prime traces and nodes of
the negotiation.
Lemma~\ref{lem:bstep} roughly says that while process $q$ goes from
node $m$ to node $n$ by action $b$,  the remaining processes 
execute $u$, after which $n$ is the unique executable node.
See also Figure~\ref{fig:ut} for an illustration.

% single out nodes of the
% negotiation.
% It refers to the existence of configurations $I(n)$ that enable
% precisely node $n$, as explained in Theorem~\ref{th:I}.
% The situation of the lemma is presented in the left part of Figure~\ref{fig:ut},
% action $b$ in the picture is an action enabled in $n$.
% The lemma says that if from an execution we cut out a trace-suffix $t$ that is
% co-prime, then the remaining part determines a node of $\Nn$ uniquely. 

\begin{lemma}\label{lem:dom}
  If $ut \in L$ and  $t$ is co-prime then:
  \begin{itemize}
  \item $\Cinit\act{u}C$ is an execution of $\Nn$ with $C=I(n)$ for some node
  $n$, and $\ndom(n)=\dmin(t)$.
\item If $ut' \in L$ for some $t'$ then $t'$ is also co-prime, and
  $\dmin(t')=\dmin(t)$. 
  \end{itemize}
\end{lemma}

% The next lemma gives some intuition behind the concept of $(b,p)$-steps.
% These are traces that fill out the part of an execution between successive
% $I$-configurations. 
% The situation of the lemma is presented in the right part of
% Figure~\ref{fig:ut}, where $c,b$ are outgoing actions of nodes
% $m$ and $n$, respectively, and $s$ is a $(c,q)$-step.

\begin{lemma}\label{lem:bstep} 
  Let  $\Cinit\act{*} I(m)\act{u}I(n)\act{*} \Cfin$
  be an execution of $\Nn$.
  We have  $m\act{(b,p)} n$  if  and only if  $u$ is a $(b,p)$-step.
\end{lemma}

% \begin{definition}
% 	Consider a sequence $wv$ with $v$ prime.
% 	We say that $u$ is $(b,p)$-step for $wv$ if
% 	\begin{itemize}
% 		\item $uv$ is prime, $b=\min(uv)$,
% 		\item $w\tequiv w'uv$,
% 		\item $p\in\dom(b)\cap\dom(\min(v))$ and $p\not\in\dom(u-b)$.
% 	\end{itemize}
% \end{definition}

% \textbf{Observation:} $u$ is uniquely determined by $w,b,p$. It may be that $u$
% does not exist.

% \begin{definition}
% 	Let $w$ be a sequence of actions, $b$ an
% 	action, $p\in\dom(b)$ a process.
% 	We say that $u$ is a $b$-predecessor of $w$  $p$, denoted
% 	$u=\pre(b,w,p,)$, if
% 	\begin{itemize}
% 		\item $u$ is trace suffix of $w$, a prime trace, and $b=\min(u)$,
% 		\item $p\not\in\dom(u-b)$.
% 	\end{itemize}
% \end{definition}

% \textbf{Observation} The notion of $u$ being $(b,p)$-step of $w$ depends on $w$
% (because $u$ is a part of it) and on the minimal action of $v$.
% It does not depend on other parts of $v$.
% Observe that $u$ is uniquely determined by $w$, $(b,p)$ and the minimal action f
% $v$.

% {\color{gray}
% \begin{lemma}[Forward-step Lemma]\label{lem:forward-step}
% 	\igw{old lemma}Suppose that $\Nn$ is a negotiation and take $uv \in L$ with $v$ non-empty and
% 	co-prime, $b=\min(t)$, and $p\in\dom(b)$.
% 	There is a $(b,p)$-step $s_p$, and a co-prime $t$ such that $v\tequiv s_pt$ and
% 	$p\in\dmin(t)$.
% 	Moreover, $s_p$ is unique for a given $p$.
% \end{lemma}
% }

The last lemma exhibits a structural property of sound deterministic
negotiations in terms of co-prime traces and $(b,p)$-steps.
\begin{lemma}[Crossing Lemma] \label{lem:cross}
	If $ws_1t_1\in L$ and $ws_2t_2\in L$, where $t_1,t_2$ are co-prime,
	$p\in\dmin(t_1)\cap\dmin(t_2)$, and $s_1,s_2$ are  
	$(b,p)$-steps, then
	\begin{itemize}
		\item $\dmin(t_1)=\dmin(t_2)$,
		\item $ws_1t_2\in L$.
	\end{itemize}
\end{lemma}
% \begin{proof}
% 	We observe that $s_1t_1$ and $s_2t_2$ are co-prime traces.
% 	By Lemma~\ref{lem:dom} we have two executions:
% 	\begin{align*}
% 		\Cinit\act{w}I(m)\act{s_1}I(n_1)\act{t_1}\Cfin\\
% 		\Cinit\act{w}I(m)\act{s_2}I(n_2)\act{t_2}\Cfin\ .
% 	\end{align*}
% 	Now, Lemma~\ref{lem:bstep} gives  $n_1=n_2=n$ because $m
%         \act{(b,p)} n_1$, $m
%   \act{(b,p)} n_2$  and $\Nn$ being deterministic.
% 	So $\dmin(t_1)=\dmin(t_2)$ by Lemma~\ref{lem:dom}.
% 	This also entails:
% 	\begin{equation*}
% 		\Cinit\act{w}I(m)\act{s_1}I(n)\act{t_2}\Cfin\,.
% 	\end{equation*}
%       \end{proof}

\subsection{The learned negotiation}\label{subsec:learning-structure}

The negotiation learned by our algorithm is built from the
following sets: 
\begin{itemize}
	\item $Q \subseteq \Act^*$ is a set of traces, we often call them \emph{nodes}.
	There should be a unique node in $Q$ that is also in
	$L$. %: $\ufin\in Q\cap L$\igw{Do we need $\ufin$?}

	\item $T \subseteq \Act^*$ is a set of co-prime traces, plus the empty trace
          $\e$.

	\item $S:Q \times \Act \times \Proc \to \Act^*$ is a partial function giving \emph{supports} for transitions: if defined,
	$S(u,b,p)$ is a $(b,p)$-step.
\end{itemize}

The use of co-prime traces for $T$ is motivated by Lemma~\ref{lem:dom}, as
runs from configurations of the form $I(n)$ are co-prime traces.
The support function is new. 
It is a generalization of the mapping $\out$ from
Sections~\ref{sec:angluin} and~\ref{sec:paths}. 
As described by Lemma~\ref{lem:bstep}, when a process $p$ executes an
action $b$ reaching a new node $n$, other processes need also to progress until
$n$ becomes the only executable node; such a progress is
a trace, and the support $S(u,b,p)$ is one such trace.

Our construction will preserve the following invariants:
\begin{description}
	\item[\Unique] For every $u,v \in Q$, $u \eqT v$ implies $u=v$.\smallskip
	\item[\PREF] For every $u\in Q$ there is $t\in T$ such that
	$ut\in L$.\smallskip
	\item[\DOM] If the support $S(u,b,p)$ is defined then $S(u,b,q)$ is defined for all $q\in\dom(b)$.\smallskip
	\item[\PREF'] If the support $S(u,b,p)$ is defined then there  exists some
	$t\in T$ with $u \, S(u,b,p)\, t \in 	L$. Moreover, if $t\not=\e$ then $p \in\dmin(t)$. \smallskip
	\item[\Closure] If the support $S(u,b,p)$ is defined then
          there is some $v\in Q$ 	with $u S(u,a,p) \eqT v$.
	% \item[\DOM2] If the support $S(u,b,p)$
  %         is defined then both conditions below are satisfied:
  %         \begin{enumerate}
  %         \item $S(u,b,q)$ is defined for all $q\in\dom(b)$.
  %         \item There  exists some $t\in T$ with $u \, S(u,b,p)\, t \in
  %           L$. Moreover, if $t\not=\e$ then $p \in\dmin(t)$. 
  %         \end{enumerate}
  \end{description}

% Let us discuss some consequences of these invariants

% \begin{lemma}\label{lem:pref-action}
% 	Suppose $(Q,T,S)$ satisfies the four invariants. 
% 	If $S(u,b,p)$ is defined then there is $t\in T$ with $u \, S(u,b,p)\, t \in
%   L$. Moreover, if $t\not=\e$ then $\min(S(u,b,p)t)=\set{b}$\igw{Changed from $p\in\dmin(t)$}.
% \end{lemma}
% \begin{proof}
% 	If $S(u,b,p)$ is defined then \Closure\ invariant gives us $v$ with $u
% 	S(u,a,p) \eqT v$. 
% 	By \PREF\ there is $t\in T$ such that $vt\in L$, hence $u S(u,a,p) t\in L$.
% 	We need to show that  $b$ is the unique minimal element of $S(u,b,p)t$.
% 	Clearly $b\in\min(S(u,b,p)t)$ as $S(u,b,t)$ is a $(b,p)$-step so $b$ is the minimal element of $S(u,b,p)$. 
% 	Suppose there is also $a\in\min(S(u,b,p)t)$. 
% 	If $a\not=b$ then $a\in \min(t)$ and $\dom(a)\cap\dom(b)=\es$.
% 	This means that the execution $\Cinit\act{u}C$ in $\Nn$ reaches a
% 	configuration $C$ where two nodes are enabled. 
% 	But by \PREF\ there is a co-prime $t'\in T$ with $ut'\in L$. 
% 	Lemma~\ref{lem:dom} says that $C=I(n)$ for some node $n$, so only $n$ is
% 	enabled in $I(n)$. 
% 	A contradiction.
% \end{proof}

 \Unique\ and \Closure\ are the basic invariants, as in
 Sections~\ref{sec:angluin} and \ref{sec:paths}.
 \DOM\ and \PREF\ are the counterparts of the invariants in
 Section~\ref{sec:paths}.
 Note that \PREF' is not a direct consequence of \PREF\ and \Closure\
 because it puts an additional condition on $\dmin(t)$.
 The next lemma shows how to restore the \Closure\ invariant once
 the other four hold.
 
\begin{lemma}\label{lem:closure}
	If a triple $(Q,T,S)$ satisfies all invariants \Unique,
        \PREF, \DOM, \PREF', \Closure, and
        $(Q,T',S')$ with $T \subseteq T'$ and $S \subseteq S'$
        satisfies all invariants  but \Closure, then Learner can extend $Q$ and restore all five 
  invariants using $O(|S|(|T'\setminus T|)+(|S'\setminus S|)|T'|)$  membership queries.
      \end{lemma}
% \begin{proof}
% 	Suppose that for some $u \in Q$ and $S(u,b,p)$ there is no $v \in Q$ with $u \,
% 	S(u,b,p) \eqTp v$. 
% 	Add $u \, S(u,b,p)$ to $Q$ and make $S(u \, S(u,b,p))$
%         undefined for all actions. 
% 	Observe that the invariants are preserved, in particular, $u \, S(u,b,p)$
% 	satisfies invariant \PREF\ because of \PREF'.

% 	Let us count the membership queries. 
% 	There are two cases.
% 	If $S(u,b,p)$ was defined, then there was some $v\in
% 	Q$ with $u \,S(u,b,p) \eqT v$. We need to ask only
% 	membership queries for $u\, S(u,b,p)t'$ and $vt'$ with $t' \in
% 	T' \setminus T$. 
% 	Otherwise, if $S'(u,b,p)$ is new we need membership queries  $u\, S'(u,b,p)t'$
% 	for all $t' \in T'$.  
% \end{proof}

From $(Q,T,S)$ satisfying all invariants we can construct the negotiation
$\tNn$ such that: 
\begin{itemize}
	\item $Q$ is the set of nodes of $\tNn$,
	\item $\ndom(u)=\dmin(t)$ if $ut\in L$ for some co-prime $t\in T$, and $\ndom(u)=\Proc$ if $u\in L$,
%	\item $\out(u)=\set{b : S(u,b,p) \text{ defined for some } p}$,
	\item $u\act{(b,p)}v$ if $S(u,b,p)$ defined and $u \,
          S(u,b,p)\eqT v$,
        \item $\ninit=\e$, and $\nfin$ is the unique node in $Q \cap L$.
\end{itemize}
Notice the use of supports in defining transitions.
We cannot simply use actions to define transitions as $T$ contains only co-prime
traces.

\begin{lemma}\label{lem:domains}
	For every $(Q,T,S)$ satisfying the invariants, the negotiation $\tNn$ is
	deterministic and satisfies the following conditions: 
	\begin{itemize}
		\item The domain $\ndom(u)$ is well-defined for every
                  node $u \in Q$. 
		\item If $S(u,b,p)$ is defined then $\ndom(u)=\dom(b)$.
		\item If $u\act{(b,p)} v$ then $p\in \ndom(u)\cap\ndom(v)$.
	\end{itemize}
\end{lemma}
Note that $\tNn$ need not be sound. 
In particular, even if $(Q,T,S)$ defines a sound negotiation, the triple
$(Q',T',S')$ obtained after an application of Lemma~\ref{lem:closure} may not be
sound. 
We will see how to restore soundness in Section~\ref{sec:making-sound}.
Before this we describe two operations that extend $T'$ and $S'$. 

% \begin{proof}
% 	Note first that the domain $\ndom(u)$ is well-defined according to
% 	Lemma~\ref{lem:dom} and invariant \PREF. 

% 	For the second statement suppose that $S(u,b,p)$ is defined. 
% 	By \PREF', $uS(u,b,p)t \in L$ for some $t \in T$ which is
%         either empty or has $p\in\dmin(t)$. 
% 	In both cases, by Lemma~\ref{lem:dom} and the definition
% 	of domains we obtain $\ndom(u)=\dmin(S(u,b,p) t)=\dom(b)$.
  
% 	For the last statement, the transition $u\act{(b,p)} v$ entails $u\,
% 	S(u,b,p)\eqT 	v$. 
% 	Moreover, by \PREF' there is $t\in T$ with $u\,
% 	S(u,b,p)\, t\in L$ and either $t=\e$ or $p\in \dmin(t)$.
% 	Hence $vt\in L$, so $p\in\ndom(v)$ holds in both cases % since $S(u,b,p)t$ is
%         %co-prime and
%         by the definition of node domains.
% 	We also have $p\in\ndom(u)$ by the second statement of the lemma.
% \end{proof}

\subsection{Two operations to extend $\tNn$}\label{sec:two-operations}

In response to an equivalence query \Teacher\ may give a counter-example that
\Learner\ then analyses in order to extend $\tNn$.
This is described in Sections~\ref{sec:negative},
\ref{sec:positive} that follow.
Here we present two operations used in these sections to actually extend $\tNn$.

\paragraph{\textbf{\OUTINC $(u,r)$}}
Suppose that we have $u \in Q$, $r$ co-prime  with $ur \in L$, but
$\min(r) \notin \out(u)$.
Since $ur \in L$ we know that $\dmin(r)=\ndom(u)$ by
Lemma~\ref{lem:dom}.
Let $a=\min(r)$. 
For every process $p \in\dom(a)$, consider the decomposition $r=ar'r_p$, where
$p\not\in\dom(r')$, and $r_p$ is either
the co-prime trace with $p \in \dmin(r_p)$, or $r_p=\e$.
We set $S(u,a,p)=ar'$.  
Since we do it for all $p\in\dom(a)$, invariant \DOM\ holds.
We add $r_p$ to $T$ to satisfy invariant \PREF'.  %\anca{$|T|\le|\de|$}
This way we restore invariants \Unique, \PREF, \DOM, and \PREF'.
The \Closure\ invariant can be restored by Lemma~\ref{lem:closure}.

% \igwin{The rest of this paragraph should go somewhere else}
% If \Closure\ does not hold anymore
% then we use Lemma~\ref{lem:closure} to restore
% it. % \ancain{either 1 MQ if
%   % old $v$ existed, or $|T|$ if none existed; total MQ: $|T||\de|$}
% Observe that Lemma~\ref{lem:closure} extends only the $Q$ component of
% $(Q,T,S)$. 
% As explained in the algorithm section, we can rearrange this case so we first add all
% transitions needed to have a run on $w$, and only afterwards we apply
% Lemma~\ref{lem:closure}. Note that Lemma~\ref{lem:closure} must be
% applied to all transitions, not only new ones, because $T$ changed. 

\paragraph{\textbf{\Target$(u' \act{(b,p)} u, r)$}}
Assume we have a transition $u' \act{(b,p)} u$ of $\tNn$ and a co-prime trace
$r$  such that $u'S(u',b,p)r\in L\not\iff ur\in L$.
Note that $p \in \ndom(u)$ because of $u' \act{(b,p)} u$ and
Lemma~\ref{lem:domains}. 
Also, $p\in\dmin(r)$ because either $ur \in L$ and $p\in\ndom(u)$, or
$u'S(u',b,p)r \in L$ and \PREF'.
We add $r$ to $T$. 
Clearly all the invariants but \Closure\ continue to hold. 
Since \Closure\ does not hold, applying Lemma~\ref{lem:closure} will
add at least one new node to $Q$. 
Afterwards all invariants are restored.

\medskip

We end with a very useful lemma allowing to detect the \Target\ case.
\begin{lemma}\label{l:path}
Let $\e \act{a_1,p_1} u_1 \act{a_2,p_2} \cdots \act{a_k,p_k} u_k$ be
a local path in $\tNn$, and $s_i=S(u_{i-1},a_i,p_i)$ be the
support of the $i$-th transition.
Let also $r$ be a co-prime trace such that $p_k\in\dmin(r)$ and $u_k r  \in L
\not\iff s_1 \dots s_kr \in L$.
There exists some index $i$ such that
\[u_{i-1} \, s_i \dots s_k r \in L \not\iff
  u_i \, s_{i+1} \dots s_k r \in L
  \]
Moreover $u_{i-1} \act{a_i,p_i} u_i$ together with $s_{i+1}\dots s_kr$ is an
instance of the \Target\ case.
Such an index $i$ can be found with $O(\log(k))$ membership queries. 
\end{lemma}
\begin{proof}
By assumption, $u_k r \in L \not\iff s_1 \dots s_k r\in
  L$ for $r$ co-prime with $p_k \in \dmin(r)$.
  Setting $u_0=\e$ we see that we cannot have $u_{i-1} \, s_i \dots
  s_k r \in L \iff u_i \, s_{i+1} \dots s_k r \in L$ for all
  $i=1,\dots, k$.
  Finding such an $i$ is done with binary search.

  In order to have get a \Target\ case we need to verify that $s_{i+1}\dots
  s_kr$ is co-prime. 
  Recall that each $s_i$ is co-prime with minimal element $a_i$. 
  Since $a_i$ and $a_{i+1}$ have a process in common, $a_{i+1}$ is after $a_i$ in
  $s_is_{i+1}$, hence all elements of $s_{i+1}$ are after $a_i$ in $s_is_{i+1}$.
  Repeating this argument we obtain that $s_i\dots s_k$ is co-prime. 
  Finally, $s_i\dots s_kr$ is co-prime because $r$ is co-prime and
  $p_k \in\dmin(s_k) \cap \dmin(r)$. % the minimal
%   action of $r$ is after the minimal action of $s_k$.
\end{proof}

\begin{corollary}\label{cor:path}
  Let $\e \act{a_1,p_1} u_1 \act{a_2,p_2} \cdots \act{a_k,p_k} u_k$ be
a local path in $\tNn$, and $s_i=S(u_i,a_{i+1},p_{i+1})$ be the support of the $i$-th transition. 
  If $u_k\not\eqT s_1\dots s_k$ then with $O(\log(k))$ queries one can find $u_i
  \act{a_i,p_i} u_{i+1}$ and $s_{i+1}\dots s_kr$ forming an instance of the
  \Target\ case. 
\end{corollary}

\subsection{Handling a negative counter-example}\label{sec:negative}
Suppose  Teacher replies to an equivalence query with a negative counter-example
to the equivalence between $\Nn$ and $\tNn$:
\begin{equation*}
  w\in \tL\setminus L\ .
\end{equation*}
We show how to find a \Target\ case with $O(\log(|w|))$
membership queries. 
% We assume that $\tNn$ is sound at this stage. \anca{do we need this?}
% This will be invariant of our algorithm.

Let $v_1$ be the longest prefix of $w$ executable in $\Nn$.
Let us suppose first that $v_1=w$, so $\Cinit \act{w} C \not=\Cfin$ in
$\Nn$.
Since $\Nn$ is sound there must exist some action $a$ executable in
$C$.
Chose some $p \in \dom(a)$ and consider the projection $w|_p=a_1 \dots
a_k$.
In $\tNn$ we have a local path $\e \act{a_1,p} u_1 \dots \act{a_k,p}
u_k=u$ and $u \in L$ by assumption ($w \in \tL$).
Let $s_i=S(u_{i-1},a_i,p)$ be the support of the $i$-th transition.
If $u \eqT s_1 \dots s_k$ then $s_1 \dots s_k \in L$.
But this is impossible, as $(s_1 \dots s_k)|_p = w|_p$, so $\Cinit
\act{s_1 \dots s_k} C'$ with $C'(p)=C(p)$.
So $u \not\eqT s_1 \dots s_k$ and we obtain by Corollary~\ref{cor:path} 
an instance of the \Target\ case, after adding one trace to $T$.

Assume now that $w=v_1 b v_2$, and chose some $p \in \dom(b)$.
Consider the projection  $v_1|_p=a_1 \dots a_k$ and the local path $\e
\act{a_1,p} u_1 \dots \act{a_k,p} u_k \act{b,p} u'$ in $\tNn$.
Let also $s_i=S(u_i,a_{i+1},p)$, and set $u :=u_k$.
By the invariants of $\tNn$ there are some $t,t' \in T$ with $ut \in L$, $u't'
\in L$.
Also, we have $u S(u,b,p) \eqT u'$, so $u S(u,b,p)t' \in L$.
Suppose that  $u \eqT s_1 \dots s_k$ holds.
Then $s_1 \dots s_k$  is executable in $\Nn$ because of $s_1 \dots s_n t \in
L$.
Consider now $t'':=S(u,b,p)t'$ and observe that $s_1 \dots s_k t'' \notin L$: if $\Cinit \act{v_1} C$ and
$\Cinit \act{s_1 \dots s_k} C'$ then $C(p)= C'(p)$, so action $b$ is
impossible in $\Nn$  after executing $s_1 \dots s_k$.
Therefore we have $u t'' \in L \not\iff  s_1 \dots s_kt'' \in L$.
We can conclude by applying Lemma~\ref{l:path} to the local path $\e
\act{a_1,p} u_1 \dots \act{a_k,p} u_k$ and $t''$, obtaining  an
instance of the \Target\ case.

\subsection{Handling a positive counter-example}\label{sec:positive}
Consider now the case where Teacher provides a positive counter-example:
\begin{equation*}
  w\in L\setminus \tL
\end{equation*}
Compared to negative counter-example case, here we need to assume that
$\tNn$ is sound, in order to be able to use the Crossing
Lemma~\ref{lem:cross}. 
We can show that \Learner\ can determine an instance either of
\OUTINC\ or of the \Target\ situation with $O(\log(|w|)$ membership
queries. 
The details can be found in the appendix.

\subsection*{Making $\tNn$ sound}\label{sec:making-sound}
Making $\tNn$ sound is important for two reasons.
The first one is that we use the soundness of $\tNn$ when handling positive counter-examples.
The second reason is that if \Learner\ asks  equivalence queries only when $\tNn$ is
sound, then \Teacher\ can answer them in \PTIME, according to
Cor.~\ref{cor:equiv}. 

After handling  counter-examples $\tNn$ is extended as described for the cases
\OUTINC\ and \Target\ in Section~\ref{sec:two-operations}.
These do guarantee that the result satisfies the invariants, but do not guarantee
that the result is sound.
In the proposition below we show how $\tNn$ can be made sound by
\Learner\ using only
membership queries.  

We assume that $\tNn$ satisfies all the invariants of
Section~\ref{subsec:learning-structure}. 
% For $p \in \Proc$ we say that a local path is a $p$-path if all
% its transitions are labeled by some $(a,p)$, with $a \in Act$.

A local path in $\tNn$, $\pi= (a_1,p_1) \dots (a_k,p_k)$ determines nodes through which it
passes $\e \act{a_1,p_1} u_1 \act{a_2,p_2} \cdots
\act{a_k,p_k} u_k$  in  $\tNn$.
We write $S(\pi)$ for the trace
$S(u_0,a_1,p_1) \dots S(u_{k-1},a_k,p_k)$ concatenating the supports
of the transitions of $\pi$.
As we have observed in Lemma~\ref{l:path} this trace is co-prime.
We say that $\pi$ as above is a \emph{$p$-path} if $p_i=p$ for $i=1,\dots,k$.

\begin{proposition}\label{prop:making-sound}
  \Learner\ can check in \PTIME\ if $\tNn$ is sound.
If the answer is no, then  \Learner\ can find either an
instance of \OUTINC\ or of \Target,  with $O(s|T|+ \log(m))$
membership queries.
\end{proposition}
\begin{proof}
  We assume throughout the proof that $\Nn$ is minimal.
  Checking whether a deterministic negotiation is sound is an
  \NLOGSPACE-complete problem~\cite{EKMW18}.
  A negotiation is not sound if and only its graph contains one of the following
  patterns: 
  \begin{enumerate}
  \item[F:] A local path from $\ninit$ to some node $n$, action $a \in \Act$, two nodes
    $n_1,n_2$  and two processes $p_1,p_2$ such that 
    \begin{itemize}
\item $\set{p_1,p_2} \subseteq \dom(n) \cap \dom(n_1) \cap \dom(n_2)$;
\item for $i=1,2$ there exists a $p_i$-path $\pi_i$ from node
  $\d(n,a,p_i)$ to node $n_i$; and 
\item $\pi_1$ and $\pi_2$ are disjoint. %, i.e., they have no node in common.
\end{itemize}
\item[C:] A local path which is a cycle and has no node $n$ on it with
  $\dom(n)$ containing all processes occurring in the cycle; moreover this cycle
  is reachable.
  \item[B:] A node that is reachable from $\ninit$ by a $p$-path, but
   has not  $p$-path to  $\nfin$.
  \end{enumerate}

Assume first that \Learner\ finds some  pattern of type F (fork) in
$\tNn$.
This means that she finds some words $u,u_1,u_2 \in Q$ with $u_1
\not= u_2$, $\set{p_1,p_2}
\subseteq \ndom(u) \cap \ndom(u_1) \cap \ndom(u_2)$, and local
paths $\pi, \pi_1,\pi_2 \in \Adom^*$ with $\e \act{\pi} u
\act{\pi_i} u_i$, and $\pi_i \in a_{p_i} \Adom^*$,  for $i=1,2$.
Moreover, every support in $S(\pi_1)$ is a $(b,p_1)$-step for some
$b$, and every support in $S(\pi_2)$ is a $(c,p_2)$-step, for some
$c$. 

Consider the local paths $\e \act{\pi} u \act{\pi_i} u_i$. 
For every prefix $\pi'_i$ of $\pi_i$ \Learner\ verifies if $u'_i\eqT S(\pi\pi'_i)$.
If it is not the case then using Cor~\ref{cor:path} she can find an
instance of the \Target\ case with $O(\log(s))$ membership queries, where $s$ is the
size of $\Nn$  ($s$ bounds the lengths of the paths
$\pi\pi_1,\pi\pi_2$).
The overall number of membership queries here is $O(s|T|+\log(s))$,
accounting for all prefixes. 

We show that the remaining case is impossible.
Towards contradiction suppose $u'_i \eqT S(\pi \pi'_i)$ for all prefixes
$\pi'_i$ of $\pi$ and $i=1,2$.
By invariant \PREF, both $S(\pi \pi_1)$ and $S(\pi \pi_2)$ are executable in $\Nn$.
%In particular, $S(\pi) a$ is executable in $\Nn$.
Since every support $S(u,b,p)$ is a $(b,p)$-step the trace $S(\pi)$
induces the local path $\pi$ in $\Nn$ from $\ninit$ to some node $n$ with
outcome $a$ and both $p_1,p_2$ in its domain (because $S(\pi)a$ is
executable in $\Nn$). 
Similarly, $S(\pi\pi_1)$ induces the local $p_1$-path $\pi_1$ in $\Nn$, from $n$ to some 
node $n_1$ with both $p_1,p_2$ in its domain (because of $\set{p_1,p_2}
\subseteq \ndom(u_1)$ and the \PREF\ invariant applied to $u_1 \in
Q$). 
Same applies to $S(\pi\pi_2$): it induces the local  $p_2$-path
$\pi_2$ in $\Nn$, from $n$ to some 
node $n_2$ with both $p_1,p_2$ in its domain.
The two paths $\pi_1,\pi_2$ are disjoint because the corresponding
nodes in $\tNn$ are $\eqT$-inequivalent and $\Nn$ is minimal.
Since $\Nn$ is sound this implies $n_1=n_2$, therefore $S(\pi\pi_1)
\eqL S(\pi\pi_2)$, so in particular $S(\pi\pi_1)
\eqT S(\pi\pi_2)$.
We obtain a contradiction to $u_1 \not\eqT u_2$, using our assumption $u_1 \eqT S(\pi \pi_1)$ and
$u_2 \eqT S(\pi \pi_2)$. 

The two remaining cases, for (C) and (B) patterns, are presented in
 the appendix.
\end{proof}

\subsection{Learning algorithm}\label{sec:learning-executions-algorithm}
We assemble all the components presented until now into a learning algorithm.
We assume there is an external call $\EQ(\tNn)$ giving \Teacher's
answer to the equivalence query $L(\tNn)\stackrel{?}{=}L(\Nn)$. 
The answer can be either $\true$ or a pair of a form
$(pos,w)$, $(neg,w)$. 
In the latter case $w$ is a counter-example to the equivalence and the first
component indicates if this counter-example is positive or negative.
Counter-examples are handled by procedure $\BinS(\ans,Q,T,S)$.
It does a binary search on a counter-example and returns an instance of \OUTINC\
or \Target, as described in Sections~\ref{sec:negative} and~\ref{sec:positive}.
The result of $\BinS(\ans,Q,T,S)$ is either a tuple $(abs,u,r)$ for which
\OUTINC$(u,r)$ holds, or a tuple $(mt,u_1,b,p,u_2,r)$ for which
\Target$(u_1,b,p,u_2,r)$ holds.  
The procedures $\OUT$ and $\TRG$ extend $(Q,T,S)$ as described in
Section~\ref{sec:two-operations}.
Then procedure $\CLOS$ restores invariant \Closure\ as described in
Lemma~\ref{lem:closure}. 
Finally, $\isSound(\tNn)$ checks if $\tNn$ is sound; if not, it either
returns an instance
of  \OUTINC\ ($\res=(abs,u,r)$)
or of \Target\ ($\res=(mt,u_1,b,p,u_2,r)$),
as described in Sections~\ref{sec:making-sound}.

\begin{algorithm}%[H]
	% \TitleOfAlgo{Learning sound negotiations by querying executions} 
	 %\KwIn{}
	 %\KwOut{}
 Init: $\ans \larr \EQ(\Nn_\es$)\;
 \lIf{($\ans=\true$)}{\Return $\Nn_\es$}
 $(Q,T,\out) \larr (\set{\e}, \set{w}, S=\text{empty function})$ \tcp*{}
 $\OUT(\e,\ans.w,Q,T,S)$ \tcp*{}
 $\CLOS(Q,T,S)$ \tcp*{restore \Closure}
	 %\BlankLine
 \While{($\ans\not=\true$)}{
 $\tNn\larr \Negotiation(Q,T,S)$ \tcp*{build $\tNn$}
 $\ans \larr \EQ(\tNn)$ \tcp*{ask \Teacher}
 \lIf(\tcp*[f]{if OK, stop}){($\ans=\true$)}{\Return $\tNn$}
 $\res \larr \BinS(\ans,Q,T,S)$ \tcp*{process}
 \Repeat{$\res\not=\true$}{
    \lIf{$\res=(abs,u,r)$}{$\OUT(u,r,Q,T,S)$}
    \lIf{$\res=(mt,u_1,b,p,u_2,r)$}{$\TRG(u_1,b,p,u_2,r,Q,T,S)$}
    $\CLOS(Q,T,S)$ \tcp*{restore \Closure}
    $\res\larr \isSound(\tNn)$
 }(\tcp*{$\tNn$ sound})
 }
 \caption{Learning algorithm with membership queries about executions.}
 \end{algorithm}
 \medskip

The set of $T$ of test traces is extended by $\OUT$ and $\TRG$,
by one for each new transition and each new state, respectively.
Thus, $|T| \le |Q| + |S|$.
Because in each iteration of the while-loop either $Q$ or $S$ is extended, the number of equivalence queries is at most $|Q|+|S|$.\igw{modified}
As in previous sections, to simplify the complexity bound we use just one parameter
$s$ for the size of the negotiation, namely the sum of the number of nodes and
the number of transitions.
By $m$ we denote the maximal size of counter-examples.

For the membership queries we observe that:
\begin{itemize}
\item $\CLOS$ uses overall $|T||S|\in O(|S|^2)$, so $O(s^2)$ membership queries 
  (see Lemma~\ref{lem:closure}).
\item  Handling a counter-example $w$ uses each $O(\log(|w|)$
  membership queries, so overall $O(s \log(m))$.
\item Making $\tNn$ sound uses $O(s|T|+\log(m))$ membership
  queries. 
So the overall number here is $O(s(s^2+ \log(m)))$.
\end{itemize}

We summarize the developments of this section in the following theorem.

\begin{theorem}\label{thm:learning-executions}
   Algorithm~2 actively learns sound deterministic negotiations,
  using membership queries on executions and equivalence queries returning
  executions. It can learn a negotiation of size $s$ 
  using $O(s(s^2+ \log(m)))$ membership queries and $s$ equivalence
  queries, where $m$ is the maximal length of  counter-examples.
\end{theorem}

The complexity bound for this algorithm is roughly by a factor $s$ bigger than that of 
Angluin's algorithm for finite automata. 
This increase is due to the part making $\tNn$ sound.
Observe though that each time algorithm makes $\tNn$ sound, it adds at least one
state or one transition, so the number of equivalence queries
decreases.

%%% Local Variables:
%%% mode: latex
%%% TeX-master: "m"
%%% End:

\section{Conclusions}\label{sec:conclusions}

We have proposed two algorithms for learning sound deterministic negotiations. 
Due to concurrency, negotiations can be exponentially smaller than equivalent
finite automata. 
Yet the complexity of our algorithms, measured in the number of queries, is
polynomial in the size of the negotiation, and even
comparable\igw{changed} to that of learning algorithms for finite automata.

An immediate further work is to implement the algorithms. 
In particular, we have not discussed how to implement equivalence queries in our
active learning algorithms. 
If Teacher has a negotiation given explicitly then the equivalence query can be done
in \PTIME\igw{added}. 
In more complicated cases this task is closely related to conformance
checking~\cite{DorEl-Maa10}, a field 
developing methods to check if a system under test conforms to a given
model.
Examples of ingenious ways of implementing the equivalence test can be found in~\cite{SmeMoeVaa15}.
Extension of these methods to distributed systems, such as negotiations, is an
interesting research direction. 

% Finally\igw{remove this paragraph?}, it would be very interesting to have means of approximating distributed
% systems by, preferably small, sound deterministic negotiations. 
% Observe that every finite state distributed system can be represented as a
% trivial sound distributed negotiation, namely a finite automaton for all the
% interleavings of the system. 
% It is intriguing to see how far we can go with the techniques hinted in the
% example from Figure~\ref{fig:neg-barrier}.

%%% Local Variables:
%%% mode: latex
%%% TeX-master: "m"
%%% End:

%\input{experiments}

\bibliographystyle{ACM-Reference-Format}
\bibliography{learning.bib}

%\input{barrier}

% Appendix
\appendix
\section{Appendix}

\subsection{Angluin's $L^*$ algorithm for DFA}

The first case is when there is no run of $\tAa$ on $w$. 
This can only happen when $w \in L \setminus L(\tAa)$.
Let $w=w' b w''$, such that $\e\act{w'} u$ and there is no $b$-transition from
$u$ in $\tAa$.
Learner adds  $b$ to $\out(u)$.
Then she checks if $ub \eqT v$ for some $v \in Q$.
If this is not the case, she adds  $ub$ to
$Q$, setting $\out(ub)=\es$.
Note that both invariants are preserved.
% , and each time 
%  $Q$ grows, $|T|\le |Q|$ membership queries are used.
The steps from this case can be repeated until $w$ has a run in $\tAa$.
The overall number of membership queries used in this step is at most
$|\out| \cdot |T|$.

The second case is the same as for Angluin learning of complete DFA: there is a
run of $\tAa$ on $w$, but  $w$ belongs to the symmetric difference of $L(\Aa)$
and $L(\tAa)$.
Assume that the run of $\tAa$ on  $w=a_1\dots a_m \in A^*$ is:
\begin{equation*}
  \e\act{a_1} u_1\act{a_2}u_2\dots\act{a_m}u_m\ .
\end{equation*}
Since $(w \in L) \not\iff (u_m \in L)$ there exists some $0 \le i
<n$ such that $u_i a_{i+1}\dots a_m \in L \not\iff u_{i+1}a_{i+2}\dots
a_m\in L$.
Such an index $i$ can be found with binary search, so that
 $O(\log(m))$ membership queries are required.
Learner adds $a_{i+2}\dots a_m$ to $T$.
Now $u_ia_{i+1} \not\eqT u_{i+1}$, so $u_ia_{i+1}$ is added to $Q$.
Setting $\out(u_ia_{i+1})=\es$ restores the \Unique\ invariant.

The \Closure\ invariant is also easy to restore.
Suppose after adding $a_{i+2}\dots a_m$ to $T$ for some $u\in Q$, and $a\in\out(u)$ there is
no $v\in Q$ with $ua\eqT v$. 
In this case add $ua$ to $Q$, and set $\out(ua)=\es$.
This operation does not invalidate \Unique.
Restoring \Closure\  after adding one element to $T$ requires $|\out|$
membership queries since for every $u\in Q, a \in\out(u)$ there 
is a unique possible $v$ to check, the one that was suitable before
extending $T$.

% The final step consists in restoring \Closure. 
% This is necessary when for some $u\in Q$, and $a\in\out(u)$ there is
% no $v\in Q$ anymore with $ua\eqT v$. 
% In this case  $ua$ is added to $Q$, setting $\out(ua)=\es$.
% This step requires $|\out|$ \anca{for complete Angluin this is
%   $|Q||A|)$} membership queries since for every $u, a \in\out(u)$ there
% is a unique possible $v$ to check, the one that was suitable before
% enlarging $T$ by one word.

The algorithm terminates, as in each case either  $Q$ or
$\out$ grows.
Note that $|T|\le |Q|$ since $T$ is extended only in the second case, where $Q$
is extended too. % and that $|Q|+|\out|$ is the size of the target DFA.
The algorithm ends with $\tAa$ being the minimal DFA of $L$. 
In total it uses at most $O(|Q|(|\out|+\log(m))$ membership queries, where
$m$ is the maximal length of counter-examples given by Teacher.
The number of equivalence queries is bounded by $|Q|+|\out|$.
Note that Angluin's algorithm for complete DFA uses $|Q|$ equivalence queries and
$O(|Q|(|Q||A|+\log(m)))$ membership queries~\cite{TTT14}.
However, the size $|out|$ of the target DFA may be much smaller than $|Q||A|$,
in particular if $A$ is very large compared to the maximal out-degree of states.

\subsection{Missing proofs from Section~\ref{sec:minimization}}

\textbf{Proposition~\ref{prop:mindfa-negotiation}}
 Let $\Nn$ be a sound deterministic
  negotiation and $\Aa$ the minimal DFA accepting
  $\paths(\Nn)$.
  Then $L(\Nn)=L(\Nn_\Aa)$. Moreover $\Nn_\Aa$ is deterministic and sound.
 
\begin{proof}
  By definition $\Nn_\Aa$ is deterministic. 

  We show first $L(\Nn) \subseteq L(\Nn_\Aa)$. Let $w \in L(\Nn)$ and suppose
  that $u \sqs w$ is the maximal trace-prefix of $w \in\Act^*$ that is executable in
  $\Nn_\Aa$.
  Let also $\Cinit \act{u} C$ in $\Nn$ and
  $\Cinit' \act{u} C'$ in $\Nn_\Aa$.
  Assume first that $|u|<|w|$ and let $a$ be the first letter after $u$ in $w$: $w = uau'$.
  Since $a$ is enabled in $C$ there is some node $n$ with $C(p)=n$ for all $p\in\dom(a)$.
  By Lemma~\ref{lem:locpath} the projection $u|_p$ of $u$ on
$p$ is a local  path in $\Nn$, from $\ninit$ to $n$.
%  Moreover $u_p$ is the $p$-maximal local path compatible with $u$. 
  % By Lemma~\ref{lem:compatible-execution-path}, $u|_p$ is a path in
  % $\Nn$ from $n_\init$ to
  % $n$. 
  This means that $(u|_p)^{-1}\paths(\Nn)=(u|_q)^{-1}\paths(\Nn)$ for all $p,q\in\dom(a)$.
  Since $\Aa$ is the minimal DFA for $\paths(\Nn)$, there is
  a state $s$ of $\Aa$ such that for every process $p\in\dom(a)$, $\Aa$ reaches
  $s$ after reading $u|_p$. 
  So in $C'$ we have $C'(p)=s$ for all $p\in\dom(a)$; by Lemma~\ref{lem:locpath}.
  Hence, $a$ is enabled in $C'$, a contradiction to the assumption
  that $a$ is not executable in $C'$.

  It remains to consider the case where $u \tequiv w$. Here we have
  that $w|_p \in \paths(\Nn)$ for every process $p$, so $w|_p$ is a
  local path in $\Nn_\Aa$ from the initial to the final node.
  This entails $w \in L(\Nn_\Aa)$ by Lemma~\ref{lem:locpath}.
  %  Let $a \in \S$ be such that
  %  $C \act{a}\act{*} \Cfin$ in $\Nn$, and let $n$ be the node with
  %  outcome $a$ enabled in $C$. By assumption, $a$ is not executable in
  %  $C'$, so there exist some $p_1,p_2 \in \dom(a)$ with
  %  $C'(p_1) \not= C'(p_2)$. The projection $u_i$ of $u$ onto $A_{p_i}$
  %  induces a local path $\pi_i$ in $\Nn$, from $\ninit$ to
  %  $C(p_1)=C(p_2)$. So $u_1^{-1} \paths(\Nn)=
  %  u_2^{-1}\paths(\Nn)=\paths(n)$. In $\Aa$ we have $q_0 \act{u_i} q_i$, with
  %  $q_1 \not=q_2$. On the other hand, $u_1^{-1} L(\Aa)=
  %  u_2^{-1}L(\Aa)$ shows that $L(q_1)=L(q_2)$, and this
  %  contradicts the minimality of $\Aa$. Thus, $L(\Nn) \subseteq
  %  L(\Nn_{\Aa})$.

  For the converse inclusion we show a stronger statement: if $w$ is an
  execution of $\Nn_\Aa$ then it is an execution of $\Nn$. 
  The statement is stronger as we do not require that $w$ is complete.

  Let $w=uav$ be an execution in $\Nn_\Aa$. 
  Suppose that $u$ can be executed in $\Nn$.
  We show that $ua$ can be executed in $\Nn$ as well. 
  We have $\Cinit\act{u} C$ in $\Nn$ and $C'_{\init}\act{u}C'\act{a}C''$ in $\Nn_\Aa$.
  Consider the projection $u|_p$ on some process $p$.
  By the definition of $\Nn_\Aa$ we have $C'(p)=\d_\Aa(s^0,u|_p)$.
  
  Since $a$ is enabled in $C'$ and by Lemma~\ref{lem:locpath}, there
  is some state $s$ of $\Aa$ with $a_q \in\out(s)$ and $\d_\Aa(s^0,u|_q)=s$ for all
  $q\in\dom(a)$. 
  % Recall that $L(s)=\es$ by the definition of $\Nn_\Aa$,
  For every $p\in\dom(a)$ consider now the node $n_p$ reached by the path $u|_p$ in $\Nn$. 
  By Lemma~\ref{lem:locpath}, $C(p)=n_p$.
  Because $u|_pa_p$ is a prefix of a complete path from $\paths(\Nn)$, we get $a\in\out(n_p)$ and $p\in\ndom(n_p)$.
  If $n_p\not=n_q$ for some $p,q\in\dom(a)$ then $C$ would be a deadlock,
  which is impossible as $\Nn$ is sound.
  Hence, $a$ is enabled in $C$.
  Finally, if $w$ is a complete execution of $\Nn_\Aa$ then it is
  complete in $\Nn$ as well, because the set of local paths is prefix-free.

  It remains to show that $\Nn_{\Aa}$ is sound. 
  Let $\Cinit' \act{u} C'\not=\Cfin'$
  in $\Nn_{\Aa}$ and assume that no action is executable in
  $C'$. 
  By the previous paragraph, we have $\Cinit\act{u} C$.
  Since $\Nn$ is sound we get   $C \act{v} \Cfin$ for some $v \in\Act^*$. 
  Once again by the above, the complete execution $uv$ of $\Nn$ gives
  us a complete 
  execution of $\Nn_\Aa$. 
  Hence $C$ is not a deadlock configuration, because $v$ can be executed from
  $C$. 
\end{proof}

\subsection{Missing proofs from Section~\ref{sec:paths}}

\textbf{Lemma~\ref{lem:lp-pos}:}
  Consider a positive counter-example $w\in L\setminus \tL$. 
  Let $v$ be the maximal trace-prefix of $w$ executable in $\tNn$. 
  So we have $\tCinit\act{v}\tC$ in $\tNn$, and no action in
  $\min(v^{-1}w)$ can   be executed from $\tC$. 
  With at most $|\Proc|$  membership queries \Learner\ can determine one of the following situations:
  \begin{description}
    \item[\OUTINC:]  An action $b\in\min(v^{-1}w)$, a node
      $u \in \Adom^*$ of $\tNn$, and a sequence $r
      \in \Act^*$  starting with $b$ such that for every $p\in\dom(b)$:
      \[ u \,
        r|_p\in L_P \quad \text{and } \quad  b_p\not\in \out(u)\,.
        \]
    \item[\NEQ:] A process $p$, and a local path $\pi \in \Adom^*$
      such that
      \[ v|_p\, \pi\in L_P \quad \not\iff \quad  u\, \pi\in L_P \quad \text{ with
        } u=\tC(p)\,.
      \]
      \end{description}

\begin{proof}
  The first possibility is that $v=w$ but $\tC$ is not a final configuration in $\tNn$.
  Since $w\not \in \tL$, for some process $p$ we have that $u=\tC(p)$
  is not the final node.
  Hence $u\not\in L_P$ while $w|_p\in L_P$ by Corollary~\ref{cor:projections}. % \igw{need some reference to
  % negotiation properties}
  We get the \NEQ\ statement of the lemma for $\pi=\e$.

  For the rest of the proof consider some $b\in\min(v^{-1}w)$.
  Since $b$ is not enabled in $\tC$ we have one of the two cases:

  \textbf{Case 1:} $\tC(p)=u$ for all $p\in\dom(b)$. 
  This is possible only when $b_p\not\in\out(u)$ for some $p$, but then by
  invariant \DOM, %\igw{make dom-completeness an invariant},
  the same holds for all $q\in\dom(b)$.
  Take $r=v^{-1}w$. 
  If $u \,r|_p\in L_P$ for every $p \in\dom(b)$ then  we get the
  \OUTINC\ statement of the lemma. %\anca{$|\Proc|$ MQ}
  Otherwise we get the \NEQ\ statement since $v|_p\, r|_p=w|_p\in L_P$.

  \textbf{Case 2:} $\tC(p)=u_p\not=u_q=\tC(q)$ for some $p,q\in\dom(b)$. 
  Hence $u_p\not\eqT u_q$ by the \Unique\ invariant.
  Let $t\in T$ be such that $u_p t \in L_P \not\iff u_q t \in L_P$. 
  Observe also that $vb$ is executable in $\Nn$ since it is a trace-prefix of $w\in L$.
  Thus there is some node $n$ such that $\ninit\act{v|_p}n$ and
  $\ninit\act{v|_q}n$ in the minimal negotiation for $L$. 
  This implies $v|_p\equiv^{L_p} v|_q$, so in particular $v|_p\eqT v|_q$.
  Hence either $u_p t \in L_P \not\iff  v|_p \, t \in L_P$ or $u_q t \in
  L_P \not\iff v|_q \, t \in L_P$. %\anca{$O(1)$ MQ}
  So we get the \NEQ\ statement of the lemma with $\pi=t$.
\end{proof}

\textbf{Lemma~\ref{lem:lp-neg}:}
  Consider a negative  counter-example $w\in \tL\setminus L$, and let $\tCinit \act{w} \tC$.
  With at most $|\Proc|$  membership queries \Learner\ can find a process $p$
      such that
      \[ w|_p \in L_P \quad \not\iff \quad  u\ \in L_P \quad \text{ for
        } u=\tC(p)\,.
      \]

\begin{proof}
  Since $w\in \tL$, all nodes in configuration $\tC(p)$ are accepting.
  By definition of $\tNn$,  for every process $p$, the node $u=\tC(p)$ is such
  that  $u\in L_P$.
  On the other hand, by Corollary~\ref{cor:projections} there is $p$ such that
  $w|_p \notin L_P$. 
  Learner can find this $p$ with at most $|\Proc|$ membership queries.
  We get $w|_p\not\in L_P$ and $u\in L_P$.
\end{proof}

Let us analyze the complexity of the learning algorithm of
Section~\ref{sec:paths}. 

The number of equivalence queries is equal to the number of iterations of the
loop.
By the above, it is bounded by the size of the negotiation (that is the sum of
the number of nodes and the number of transitions). 
Note also that $|T| \le |Q|+|\out|$, since $\OUT$ adds one element to $T$
for each new transition, respectively~$\BinS$ adds one element per call.
Let us estimate the number of membership queries.
The calls of $\OUT$ altogether make $O(|\out||T|)$ membership queries,
We can over-approximate this by $O(|\out|^2)$.
The same complexity  holds for the calls of $\CLOS$ because this
procedure checks  $u',a_p,v$ with $u'a_p \eqT v$ before
enlarging $T$, only w.r.t.~newly added words in $T$.
Finally, checking whether case \OUTINC\ or \NEQ\ holds accounts for
$|Q| \cdot |\Proc|$ membership queries.

\subsection{Missing proofs from Section~\ref{sec:executions}}

\textbf{Crossing Lemma, Lemma~\ref{lem:cross}:}
	Suppose that $\Nn$ is sound and deterministic.
	If $ws_1t_1\in L$ and $ws_2t_2\in L$, with $t_1,t_2$ co-prime,
	$p\in\dmin(t_1)\cap\dmin(t_2)$, and $s_1,s_2$ are  
	$(b,p)$-steps then
	\begin{itemize}
		\item $\dmin(t_1)=\dmin(t_2)$,
		\item $ws_1t_2\in L$.
	\end{itemize}

\begin{proof}
	We observe that $s_1t_1$ and $s_2t_2$ are co-prime traces.
	By Lemma~\ref{lem:dom} we have two executions:
	\begin{align*}
		\Cinit\act{w}I(m)\act{s_1}I(n_1)\act{t_1}\Cfin\\
		\Cinit\act{w}I(m)\act{s_2}I(n_2)\act{t_2}\Cfin\ .
	\end{align*}
	Now, Lemma~\ref{lem:bstep} gives  $n_1=n_2=n$ because $m
        \act{(b,p)} n_1$, $m
  \act{(b,p)} n_2$  and $\Nn$ being deterministic.
	So $\dmin(t_1)=\dmin(t_2)$ by Lemma~\ref{lem:dom}.
	This also entails:
	\begin{equation*}
		\Cinit\act{w}I(m)\act{s_1}I(n)\act{t_2}\Cfin\,.
	\end{equation*}
      \end{proof}

\textbf{Lemma~\ref{lem:closure}:}
	If a triple $(Q,T,S)$ satisfies all invariants \Unique,
        \PREF, \DOM, \PREF', \Closure, and
        $(Q,T',S')$ with $T \subseteq T'$ and $S \subseteq S'$
        satisfies all invariants  but \Closure, then Learner can extend $Q$ and restore all five 
  invariants using $O(|S|(|T'\setminus T|)+(|S'\setminus S|)|T'|)$  membership queries.
 
\begin{proof}
	Suppose that for some $u \in Q$ and $S(u,b,p)$ there is no $v \in Q$ with $u \,
	S(u,b,p) \eqTp v$. 
	Add $u \, S(u,b,p)$ to $Q$ and make $S(u \, S(u,b,p))$
        undefined for all actions. 
	Observe that the invariants are preserved, in particular, $u \, S(u,b,p)$
	satisfies invariant \PREF\ because of \PREF'.

	Let us count the membership queries. 
	There are two cases.
	If $S(u,b,p)$ was defined, then there was some $v\in
	Q$ with $u \,S(u,b,p) \eqT v$. We need to ask only
	membership queries for $u\, S(u,b,p)t'$ and $vt'$ with $t' \in
	T' \setminus T$. 
	Otherwise, if $S'(u,b,p)$ is new we need membership queries  $u\, S'(u,b,p)t'$
	for all $t' \in T'$.  
      \end{proof}

\textbf{Lemma~\ref{lem:domains}:}
	For every $(Q,T,S)$ satisfying the invariants, the negotiation $\tNn$ is
	deterministic and satisfies the following conditions: 
	\begin{itemize}
		\item The domain $\ndom(u)$ is well-defined for every
                  node $u \in Q$. 
		\item If $S(u,b,p)$ is defined then $\ndom(u)=\dom(b)$.
		\item If $u\act{(b,p)} v$ then $p\in \ndom(u)\cap\ndom(v)$.
	\end{itemize}

\begin{proof}
	Note first that the domain $\ndom(u)$ is well-defined according to
	Lemma~\ref{lem:dom} and invariant \PREF. 

	For the second statement suppose that $S(u,b,p)$ is defined. 
	By \PREF', $uS(u,b,p)t \in L$ for some $t \in T$ which is
        either empty or has $p\in\dmin(t)$. 
	In both cases, by Lemma~\ref{lem:dom} and the definition
	of domains we obtain $\ndom(u)=\dmin(S(u,b,p) t)$ $=\dom(b)$.
  
	For the last statement, the transition $u\act{(b,p)} v$ entails $u\,
	S(u,b,p)\eqT 	v$. 
	Moreover, by \PREF' there is $t\in T$ with $u\,
	S(u,b,p)\, t\in L$ and either $t=\e$ or $p\in \dmin(t)$.
	Hence $vt\in L$, so $p\in\ndom(v)$ holds in both cases % since $S(u,b,p)t$ is
        %co-prime and
        by the definition of node domains.
	We also have $p\in\ndom(u)$ by the second statement of the lemma.
\end{proof}

\paragraph{Handling a positive counter-example}

In this part we need to assume that $\tNn$ is sound, in order to be
able to use the Crossing Lemma~\ref{lem:cross}.

We start by taking the longest trace-prefix $v_1$ of $w$ executable in
$\tNn$. 
We get $w=v_1br_1$ and $\tCinit\act{v_1} \tC_1$ with $b$ not enabled in $\tC_1$:
there is $p\in\dom(b)$ with $u_b=\tC_1(p)$ having no outgoing transition on $b$.
The sequence $br_1$ is trace-equivalent to $v_2br_2$, with $br_2$ co-prime.
For $v=v_1v_2$ we have that $w$ and $vbr_2$ are trace-equivalent, and $vt\not\in
\tL$ for every co-prime $t$ starting with $b$.

It may happen that there is no run $\tCinit\act{v}$, so we introduce a
notation.
We write $\tCinit\dact{v_3} u$ if for the maximal executable trace-prefix of $v_4$
of $v_3$  we have $\tCinit\act{v_4} \tC$ and $\tC(p)=u$ for some $p$.
So $\tCinit\dact{v_3} u$ means that by executing the maximal possible
trace-prefix of $v_3$ some process reaches node $u$ in $\tNn$.
Thus we have $\tCinit\dact{v} u_b$, for $v,u_b$ defined in the previous
paragraph. 

If $u_bbr_2\in L$ then we are in the \OUTINC\ case.
Hence suppose $u_bbr_2\not\in L$.
We show below how \Learner\ finds an instance of \Target\ in $\tNn$.
For this we use two auxiliary lemmas:%  whose statements resemble the
% situation we have seen for the negative counter-example.

\begin{lemma}\label{lem:positive2}
  Suppose we have $u_k$, $v_k$, $t_k$ and co-prime $t$ with the following properties:
  \begin{equation}\label{eq:pos-inv2}
    \tCinit\dact{v_k} u_k\qquad u_k t_k\in L\qquad v_kt_k\not\in L\qquad
     v_kt\in L
  \end{equation}
  Then either \Target\ holds for some transition going into $u_k$ and
  trace $t_k$, or
  we can find $v_{k-1}$ shorter than $v_k$, and $u_{k-1}$, $t_{k-1}$ for which
  properties~\eqref{eq:pos-inv2} hold.
\end{lemma}
\begin{proof}
  Because $t$ is assumed to be co-prime we have  in $\Nn$ a run of the
  form $\Cinit\act{v_k}I(n_k)\act{t}\Cfin$.
  Consider the last letter of $v_k$, say $c$, and some process $p\in \dom(c)$.
  We must have $n_{k-1}\act{c,p}n_k$ for some node $n_{k-1}$ in $\Nn$.
  We get a decomposition of the above run as
  $\Cinit\act{v_{k-1}}I(n_{k-1})\act{s_k}I(n_k)\act{t}\Cfin$ with $s_k$ a
  $(c,p)$-step.
  
  In $\tNn$ we  have a corresponding transition $u_{k-1}\act{c,p}u_k$.
  If $u_{k-1}S(u_{k-1},c,p)t_k\not\in L$ then we have the \Target\
  case for this transition and $t_k$.

  So we suppose $u_{k-1}S(u_{k-1},c,p)t_k\in L$ for the rest of the proof.
  Observe that $\tCinit\dact{v_{k-1}} u_{k-1}$.
  If $v_{k-1}S(u_{k-1},c,p)t_k\not\in L$ then we get the  conclusion
  of the lemma for
  $t_{k-1}=S(u_{k-1},c,p)t_k$ and the co-prime trace  $s_k t$.
  
  We show that $v_{k-1}S(u_{k-1},c,p)t_k\in L$ is impossible.
  Observe that $p \in\dmin(t) \cap \dmin(t_k)$ because $p \in
  \dom(n_k) \cap \ndom(u_k)$.
  This allows us to apply Crossing Lemma~\ref{lem:cross} to
  $v_{k-1}S(u_{k-1},c,p)t_k\in L$ and $v_{k-1}s_kt\in L$.
  We get 
  $v_{k-1}s_kt_k=v_kt_k\in L$, contradicting the assumption~\eqref{eq:pos-inv2}.
\end{proof}

\begin{lemma}\label{lem:positive1}
  Suppose we have $v_k$, and a co-prime trace $t_k$ such that for some $u_k$:
\begin{equation}\label{eq:pos-inv}
  \tCinit\dact{v_k} u_k\qquad u_k t_k\not\in L\qquad v_kt_k\in L
\end{equation}
Then either  \Target\ holds for some transition going into $u_k$ and
trace $t_k$, or we
can find a shorter $v_{k-1}$ and some $t_{k-1}$ for which either the
conditions~\eqref{eq:pos-inv} or the conditions of Lemma~\ref{lem:positive2}
hold.
\end{lemma}
Observe that conditions~\eqref{eq:pos-inv} hold for $v_k=v$, $u_k=u_b$, and
$t_k=br_2$.
Thus the positive case will be complete by proving
Lemma~\ref{lem:positive1}:

\begin{proof}[Proof of Lemma~\ref{lem:positive1}]

  Consider the last letter of $v_k$, say $c$, and some process $q\in \dom(c)$.
  In $\tNn$ we have $u_{k-1}\act{c,q} u_k$ for some $u_{k-1}$.
  By the invariants for $\tNn$ there exists a support $S(u_{k-1},c,q)$
  such that $u_{k-1}S(u_{k-1},c,q)\eqT u_k$. 

  If $u_{k-1}S(u_{k-1},c,q)t_k\in L$ then $u_{k-1}\act{c,q} u_k$ together with
  $t_k$ forms a \Target\ case.
  
  We are left to consider $u_{k-1}S(u_{k-1},c,q)t_k\not\in L$. 
  We need to find $v_{k-1}$  with $\tCinit\dact{v_{k-1}} u_{k-1}$.
  For this we take a run in $\Nn$: $\Cinit\act{v_k}I(n_k)\act{t_k}\Cfin$.
  It exists as $v_kt_k\in L$, and $t_k$ is a co-prime trace so the intermediate
  configuration must be of the form $I(n_k)$ for some $n_k$.
  Since $c$ is the last letter of $v_k$, we have a transition
  $n_{k-1}\act{c,q}n_k$ in $\Nn$, for some $n_{k-1}$.
  Lemma~\ref{lem:bstep} allows us to decompose this run further into
  $\Cinit\act{v_{k-1}}I(n_{k-1})\act{s_k}I(n_k)\act{t_k}\Cfin$
  with $s_k$ being a $(c,q)$-step, and $v_{k-1}$ some strict prefix of $v_k$.
  We claim $\tCinit\dact{v_{k-1}} u_{k-1}$. 
  This holds as $\tCinit\dact{v_{k-1}s_k} u_{k}$ and $s_k$ is a $(c,q)$-step.
  
  If $v_{k-1}S(u_{k-1},c,q)t_k\in L$ we get properties~\eqref{eq:pos-inv} for $v_{k-1}$
  and $t_{k-1}=S(u_{k-1},c,q)t_k$.
  
  The last case is when $v_{k-1}S(u_{k-1},c,q)t_k\not\in L$.
  By invariants for $\tNn$, there is $t\in T$ such that $u_{k-1}S(u_{k-1},c,q)t\in L$.
  We claim that $v_{k-1}S(u_{k-1},c,q)t\not\in L$, giving
  us  conditions~\eqref{eq:pos-inv2} of Lemma~\ref{lem:positive2} for $t_{k-1}=S(u_{k-1},c,q)t$.
  To see the claim, suppose to the contrary that  $v_{k-1}S(u_{k-1},c,q)t\in L$.
  Since  $v_{k-1}s_kt_k\in L$ the Crossing Lemma~\ref{lem:cross} implies $v_{k-1}S(u_{k-1},c,q)t_k\in L$,
  but we have assumed the contrary. 
  \end{proof}

\paragraph{Proof of Proposition~\ref{prop:making-sound}}

We present the remaining two cases of the proof.

Assume now that \Learner\ finds some pattern C (non-dominant cycle). 
This means that $\tNn$ has some local paths $\pi_1,\pi_2$, with
$\pi_1$ from $\e$ to some node $u \in Q$, and $\pi_2$ a cycle from $u$
to $u$ with no node containing in its domain all processes on the
cycle.
If $u \not\eqT S(\pi_1 \pi_2^k)$ for some $k$ then 
\Learner\  finds an instance of the \Target\ case.
We claim that this holds for $k=s$, where $s$ is the size of $\Nn$.
Since Learner does not know $s$, she needs to repeat the equivalence test for
$k=2,4,8,\dots$. 
So she needs $\log(s)$ tests.
The overall number of membership queries here is again  $O(\log(s)|T|)$. 

It remains to prove the claim from the previous paragraph.
Assume conversely that $u \eqT S(\pi_1 \pi_2^s)$.
In particular, $S(\pi_1 \pi_2^s)$ is executable in $\Nn$ for all these $k$,
by invariant \PREF.
The trace $S(\pi_1 \pi_2^s)$ induces the local path $\pi_1 \pi_2^s$ in $\Nn$.
Let $n_i$ be the node reached by $\pi_1 \pi_2^i$ in $\Nn$ for $i=1,\dots,s$.
Let $i<j$ be the smallest indices such that $n_i=n_j$.
So we obtain a local cycle $\pi_2^{j-i+1}$ in $\Nn$ that has
no dominant node.
This contradicts the fact that $\Nn$ is sound.

The last case is where \Learner\ finds some pattern B (blocking) in $\tNn$.
So we assume that there is some $p$-path $\e \act{\pi} u$ and
$u$ has no $p$-path to the unique accepting state of $\tNn$.
If $u \not\eqT S(\pi)$ then by Corollary~\ref{cor:path} \Learner\ finds an
instance of the \Target\ case with $O(\log(s))$ membership queries
($s$ is an upper bound on the length of $\pi$).

So assume that $u \eqT S(\pi)$.
In particular, using invariant \PREF\ we infer the existence of some
$t \in T$ such that $S(\pi)t \in L$.
Since $u \notin L$, $t$ must be non-empty.
Moreover, invariant \PREF' tells us that $p \in\dmin(t)$.
Consider the decomposition $t=t_1 t_2 \dots t_k$, with $t_i \dots t_k$ all co-prime suffixes of $t$
with $p$ in the domain of the minimal action $a_i$ of  $t_i$.
%We add all $t_i \dots t_k$ to $T$.\igw{We do not do this}
Take the $p$-path $\pi'=(a_1,p) \dots (a_k,p)$.
If the $p$-path $\pi'$ is not possible from $u$ in $\tNn$ then \Learner\ finds
the 
\OUTINC\ case for some $t_i\dots t_k$.
So assume that the path $u \act{\pi'} v$ exists in $\tNn$, with $u
\act{a_1,p} v_1 \act{a_2,p} v_2 \dots \act{a_k,p} v_k=v$. 
By assumption on $u$ we have $v \notin L$.
Let $s_1,\dots,s_k$ denote the associated supports.
Recall that  $S(\pi) t_1 \dots t_k \in L$.

First we check if $S(\pi) s_1 \dots s_k \eqT v_k$.
If this is not the case, Cor.~\ref{cor:path} applies and \Learner\
finds a \Target\ case.
So we assume that $S(\pi) s_1 \dots s_k \eqT v_k$, hence $S(\pi) s_1
\dots s_k \notin L$.

Since $S(\pi) s_1 \dots s_k \notin L$ but $S(\pi) t_1 \dots t_k \in L$
there exists $i$ such that $S(\pi) s_1 \dots s_i t_{i+1} \dots
t_k \in L$ but at the same time $S(\pi) s_1 \dots s_{i+1} t_{i+2} \dots t_k
\notin L$.

We check first if $v_i t_{i+1} \dots t_k \notin L$.
If this holds, Lemma~\ref{l:path} applies to the local path $\e \act{\pi}
u \act{a_1,p} v_1 \dots \act{a_i,p} v_i$ and $t_{i+1} \dots t_k$.
So \Learner\ can again find a \Target\ case.
The same argument applies when $v_{i+1} t_{i+2} \dots t_k \in L$.

So we can assume that $v_i t_{i+1} \dots t_k \in L$ and $v_{i+1}
t_{i+2} \dots t_k \notin L$.
% Consider now for every $i$ the trace $S(\pi) s_1 \dots s_i t_{i+1}
% \dots t_k$.
% If $S(\pi) s_1 \dots s_i \not\eqT v_i$  for some $i$, then \Learner\ can find an
% instance of the \Target\ case thanks to Corollary~\ref{cor:path}. 
% This requires overall $O(m|T|+\log(m))$ membership queries, where
% $m$ is a bound on 
% the lengths of the counter-examples (elements of $T$, resp.).
%
%So we  assume that $S(\pi) s_1 \dots s_i \eqT v_i$ for all $i$. 
%In particular, we know that $S(\pi) s_1 \dots s_k \notin L$.
% Since $S(\pi) t_1 \dots t_k \in L$ there is some index $i$ such
% that $S(\pi) s_1 \dots s_i t_{i+1} \dots t_k \in L$, but  $S(\pi) s_1
% \dots s_{i+1} t_{i+2} \dots t_k \notin L$.
% Moreover, $v_i t_{i+1} \dots t_k \in L$.
We also have $v_{i+1} t' \in L$ for some $t' \in T$ by invariant
\PREF.
Since $v_{i+1} \eqT v_i s_{i+1}$ by invariant \Closure, we obtain $v_i
s_{i+1} t' \in L$.
Now we can apply Lemma~\ref{lem:cross} to $v_i t_{i+1} \dots t_k \in
L$ and $v_i s_{i+1} t' \in L$ , since $s_{i+1}$, $t_{i+1}$ are both $(a_{i+1},p)$-steps, and
$p\in \dmin(t_{i+2} \dots t_k) \cap \dmin(t')$.
We obtain $v_i s_{i+1} t_{i+2} \dots t_k \in L$.
Together with $v_{i+1} t_{i+2} \dots t_k \notin L$ we get an instance
of the \Target\ case.

Overall \Learner\ uses $O(|T|+\log(m))$ membership queries.
\igw{Count complexity}

%%% Local Variables:
%%% mode: latex
%%% TeX-master: "m"
%%% End:

\end{document}